%% file: manuscript.tex
\documentclass[12pt]{article}

\usepackage{arxiv}

\usepackage[utf8]{inputenc} 
\usepackage[T1]{fontenc}    
\usepackage{hyperref}       
\usepackage{url}            
\usepackage{booktabs}       
\usepackage{nicefrac}       
\usepackage{microtype}      
\usepackage{lipsum}		    
\usepackage{graphicx}
\usepackage{doi}
\usepackage[
    indention = 24pt
    ]{subcaption}           

\usepackage{mathtools,amsmath,amsfonts,amssymb,amsthm,subcaption}
\usepackage{xcolor}
\usepackage{enumitem}
\usepackage[ruled]{algorithm2e}
\usepackage{xparse}

\usepackage[
    backend		= biber,
    style		= ieee,
    natbib		= true,
    giveninits  = true,
    maxbibnames = 99,
    minbibnames = 1,
    url			= false,
    doi			= true,
    eprint		= false,
    isbn 		= false,
    clearlang	= true,
    language 	= english
    ]{biblatex}
\theoremstyle{plain}\newtheorem{proposition}{Proposition}[section]
\theoremstyle{plain}\newtheorem{theorem}{Theorem}[section]
\theoremstyle{plain}\newtheorem{lemma}{Lemma}[section]
\theoremstyle{plain}

\theoremstyle{definition}

\input{mathdef}           

\title{Exact local recovery for Chemical Shift Imaging}

\date{\today}	          

\author{ \href{https://orcid.org/0000-0001-8773-3722}{\includegraphics[scale=0.06]{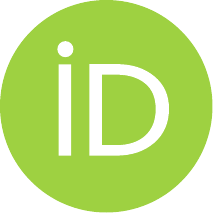}\hspace{1mm}Cristóbal Arrieta} \\
	Faculty of Engineering\\
	Universidad Alberto Hurtado\\
	Santiago, Chile \\
	\texttt{carrieta@uahurtado.cl} \\
	\And
	\href{https://orcid.org/0000-0002-2533-2509}{\includegraphics[scale=0.06]{orcid.pdf}\hspace{1mm}Carlos~A.~Sing~Long} \\
	Institute for Mathematical and \\
    Computational Engineering\\
	and Insitute for Biological \\
    and Medical Engineering\\
	Pontificia Universidad Cat\'olica de Chile\\
	Santiago, Chile \\
	\texttt{casinglo@uc.cl} \\
}

\renewcommand{\shorttitle}{}

\hypersetup{
    pdftitle    = {Exact local recovery for Chemical Shift Imaging},
    pdfsubject  = {eess.SP, eess.IV},
    pdfauthor   = {C.~Arrieta, C.~A.~Sing~Long},
    pdfkeywords = {Magnetic Resonance Imaging, Water/fat quantification, Chemical Shift Imaging, Sum of exponentials, Wirtinger flow, Signal recovery, Non-convex optimization},
}

\begin{document}
\maketitle

\begin{abstract}

    Chemical Shift Imaging (CSI) or Chemical Shift Encoded Magnetic Resonance Imaging (CSE-MRI) enables the quantification of different chemical species in the human body, and it is one of the most widely used imaging modalities used to quantify fat in the human body. Although there have been substantial improvements in the design of signal acquisition protocols and the development of a variety of methods for the recovery of parameters of interest from the measured signal, it is still challenging to obtain a consistent and reliable quantification over the entire field of view. In fact, there are still discrepancies in the quantities recovered by different methods, and each exhibits a different degree of sensitivity to acquisition parameters such as the choice of echo times.

    Some of these challenges have their origin in the signal model itself. In particular, it is non-linear, and there may be different sets of parameters of interest compatible with the measured signal. For this reason, a thorough analysis of this model may help mitigate some of the remaining challenges, and yield insight into novel acquisition protocols. In this work, we perform an analysis of the signal model underlying CSI, focusing on finding suitable conditions under which recovery of the parameters of interest is possible. We determine the sources of non-identifiability of the parameters, and we propose a reconstruction method based on smooth non-convex optimization under convex constraints that achieves exact local recovery under suitable conditions. A surprising result is that the concentrations of the chemical species in the sample may be identifiable even when other parameters are not. We present numerical results illustrating how our theoretical results may help develop novel acquisition techniques, and showing how our proposed recovery method yields results comparable to the state-of-the-art.
\end{abstract}


\newlength{\pskip}              
\setlength{\pskip}{4pt}

\linespread{1.1}

\parskip = 1pt

\parskip = \pskip

\input{body}

\printbibliography[title=References]

\appendix
\input{proofs}

\end{document}

%% file: mathdef.tex
\NewDocumentCommand{\vc}{g}{
    \IfNoValueTF{#1}
    {%
        {\boldsymbol c}
    }
    {%
        {\boldsymbol #1}
    }%
}
\newcommand{\set}[1]{\{#1\}}
\newcommand{\Lset}[1]{\left\{#1\right\}}
\newcommand{\bset}[1]{[#1]}
\newcommand{\iprod}[2]{\langle #1,\,#2\rangle}
\newcommand{\bprod}[2]{[#1,\,#2]}
\newcommand{\nrm}[1]{\|#1\|}
\newcommand{\nrmop}[1]{\nrm{#1}_{\mathrm{op}}}
\newcommand{\nrmW}[1]{|#1|_W}
\newcommand{\bmtx}[1]{\begin{bmatrix} #1 \end{bmatrix}}

\def\supp{\operatorname{\bf supp}}
\def\diag{\operatorname{\bf diag}}
\def\range{\operatorname{\bf range}}
\def\imag{\operatorname{Im}}
\def\real{\operatorname{Re}}
\def\ev{\operatorname{\bf E}}
\def\sign{\operatorname{\bf sgn}}

\def\vphi{\varphi}
\def\eps{\varepsilon}
\def\lam{\lambda}
\def\th{\theta}

\def\sums{\sum\nolimits}
\def\prods{\prod\nolimits}
\def\bigcaps{\bigcap\nolimits}
\def\bigcups{\bigcup\nolimits}

\def\adj{^{\ast}}
\def\opt{^{\star}}

\def\N{\mathbb{N}}
\def\No{\N_0}
\def\Z{\mathbb{Z}}
\def\R{\mathbb{R}}
\def\Rp{\R_+}
\def\C{\mathbb{C}}
\def\Hp{\mathbb{H}^+}

\def\vPhi{\vc{\Phi}}
\def\vPhip{\vPhi^+}

\def\xio{\xi_0}
\def\vco{\vc_0}
\def\vnull{\vc{0}}

\def\vA{\vc{A}}

\def\vT{\vc{T}}

\def\vM{\vc{M}}
\def\vMp{\vM^+}
\def\vI{\vc{I}}
\def\vJ{\vc{J}}
\def\vP{\vc{P}}

\def\vPp{\vP_{R}}
\def\vR{\vc{R}}
\def\vW{\vc{W}}

\def\cMP{\vc_{\text{MP}}}
\def\cRI{\vc_{\text{RI}}}
\def\whxi{\widehat{\xi}}
\def\whc{\widehat{\vc}}

\def\vs{\vc{s}}
\def\vso{\vs_0}
\def\tvso{\tilde{\vs}_0}

\def\vu{\vc{u}}

\def\vv{\vc{v}}

\def\vx{\vc{x}}
\def\vy{\vc{y}}
\def\vz{\vc{z}}

\def\mZ{\mathcal{Z}}

\def\Di{\Delta}
\def\vDi{\vc{\Di}}

\def\fo{f_0}

\def\tne{\tau_{n_e}}
\def\tS{\tau_S}
\def\gxio{\gamma_{\xio}}

\def\sol{\mathcal{S}}

\def\muo{\mu_{\xio}}

\def\dxi{\partial_{\xi}}

\def\Lo{L_{\xio}}

\def\epso{\eps_0}
\def\phih{\phi_{h}}

\def\Cg{C_{\phi}}
\def\ngh{\mathcal{N}}
\def\Vint{V^{\circ}}
\def\epsg{\eps_g}
\def\phio{\phi_0}

\def\HWfo{H_W^{\fo}}

\def\vo{v_0}
\def\vL{\vc{L}}

\def\LCM{\operatorname{\bf lcm}}
\def\sel{\mathsf{S}}                

%% file: body.tex
\section{Introduction}

Magnetic Resonance Imaging (MRI) is one of the most widely used biomedical imaging techniques due in part to its flexible imaging modalities. Among them, Chemical Shift Imaging (CSI)~\cite{brown_nmr_1982,shen_vivo_1999}, which is also known as Chemical-Shift-Encoded MRI (CSE-MRI)~\cite{daude_comparative_2024}, enables the quantification of the concentrations of chemical species that are known {\em a priori} to be in the human body. This modality is routinely used to quantify fat and the Proton Density Fat Fraction (MRI-PDFF) and thus most of the current understanding of this modality is derived from this application. The MRI-PDFF is a critical biomarker used to evaluate hepatic steatosis~\cite{tang_nonalcoholic_2013,byrne_time_2016,wibulpolprasert_correlation_2024} that has been validated exhaustively~\cite{park_differences_2024}.

Computing the MRI-PDFF requires an accurate estimate of the concentrations of both water and fat. The original signal model for the {\em water-fat separation problem}, which assumes a standard gradient echo acquisition, assumes that the signal at each voxel is~\cite{reeder_waterfat_2007}
\[
    s(t,\phi,c_w,c_f) = e^{2\pi i \Delta f t} c_w + e^{2\pi i \phi t}e^{2\pi i \Delta f t}  c_f
\]
where \(c_w\) and \(c_f\) are the concentrations of water and fat, \(\Delta f\) is the {\em chemical shift} of fat with respect to water, and \(\phi\) models the field inhomogeneities or {\em fieldmap}. To recover the concentrations, the signal is first sampled at different {\em echo times} using a multi-echo gradient echo sequence. Then, an iterative non-linear least-squares (NLS) method is used to solve the problem, alternating between the estimation of a linearized fieldmap for fixed concentrations, and estimating the concentrations for a fixed fieldmap. This approach is the so-called IDEAL algorithm~\cite{reeder_multicoil_2004}.

This signal model can be generalized to an arbitrary number of species, or to more complex models for the resonance signal of each species. In the literature, these generalizations have been studied using {\em in silico}, {\em in vitro} and {\em in vivo} data. These studies often focus on fatty liver disease, and highlight the challenges of ensuring an accurate and consistent fat quantification over the entire field of view. This challenge is in large part due to the signal model itself. As can be seen, the concentrations and the fieldmap cannot be uniquely identified from the signal and, as a consequence, the solution found by an iterative method depends critically on the initial condition. In practice this implies that there are {\em water-fat swap artifacts} in which the concentrations of water and fat are incorrectly assigned to the other.

Many modifications have been proposed to avoid this, from improvements to the signal model to improvements on the recovery method. For instance, Hamilton {\em et al.}~\cite{hamilton_vivo_2011} made an important contribution in characterizing the liver triglyceride spectrum {\em in vivo}, using MR spectroscopy. This enables using a {\em multi-peak} fat spectrum model~\cite{yu_multiecho_2008}
\[
    s(t,\phi,c_w,c_f) = e^{2\pi i \phi t} c_w + e^{2\pi i \phi t} \left(\sums_{p=1}^{n_p} w_p e^{2\pi i \Delta f_p t}\right) c_f 
\]
for \(w_1,\ldots, w_{n_p} \geq 0\)  with \(w_1 + \ldots + w_{n_p} = 1\), and where \(\Delta f_1,\ldots, \Delta f_{n_p}\) and \(w_1,\ldots, w_{n_p}\) are the known frequencies and relative amplitudes of the peaks of the fat spectrum. In practice, using a multi-peak fat model significantly improves the accuracy of the MRI-PDFF compared to the {\em single-peak} model~\cite{hong_mri_2018,simchick_fat_2018}. More recently, it has been concluded that it is essential to use a multi-peak fat model to obtain an accurate fat quantification, but that the benefits decrease as the number of peaks increase, leading to equivalent MRI-PDFF~\cite{wang_quantification_2017}. Nowadays the most widely accepted fat model is a 6-peak triglyceride model~\cite{hamilton_vivo_2011}. Another improvement in the signal model is the accounting of the confounder \(R_2^*\) inducing signal decay. Correcting for this effect can improve the accuracy of MRI-PDFF. The decay due to \(R_2^*\) can be modeled by replacing the fieldmap \(\phi\) by a complex variable  \(\xi\)~\cite{yu_multiecho_2007}. Finally, another confounder, the \(T_1\) signal weight, can be minimized by using a small flip angle for the acquisition~\cite{hu_ismrm_2012}. The current clinical guidelines also include the use of a 6-echo acquisition, and they have standardized requirements for the first echo and echo spacing that depend on the main field strength~\cite{henninger_practical_2020}. In contrast, an accurate and precise measurement of \(R_2^*\) to quantify iron concentration requires a first echo and echo spacing of about 1ms at 1.5T. Therefore, a 6-echo sequence is not enough to obtain
reliable \(R_2^*\) measurements, and 8 to 12-echoes have been recommended~\cite{heba_accuracy_2016}. All these guidelines have been elaborated on the basis of empirical and practical experiences.

Other approaches to avoid artifacts rely on improved recovery methods. The first such algorithm was Region-Growing IDEAL, which sorts image voxels using a spiral trajectory, starting from a reliable voxel, to ensure that the final fieldmap is smooth and consistent~\cite{yu_field_2005}. Hernando {\em et al.} introduced a formulation based on variable projection, objective discretization, and a graph-cut algorithm~\cite{hernando_robust_2010}. Tsao and Jiang introduced a formulation based on a multi-scale decomposition of the signal, solving the problem hierarchically and achieving robust results with a fast solver~\cite{tsao_hierarchical_2013}. The two last algorithms have partially inspired the use of regularization for the fieldmap, although this may lead to over-smoothed fieldmaps. The JIGSAW algorithm introduced a local smooth fieldmap, but the model only considers a single-peak fat spectrum~\cite{wenmiao_lu_jigsaw_2011}. FLAME imposes two different models for a fat or water dominant pixels, which are then combined with spatial fieldmap smoothing~\cite{yu_robust_2012}. Max-IDEAL estimates the field map using a convex relaxation of the model and a spatial filter~\cite{soliman_maxideal_2014}. B0-NICE avoids regularizing the fieldmap, using phase unwrapping instead~\cite{liu_method_2015}. In~\cite{berglund_multiscale_2017} a quadratic optimization graph-cut is combined with a multi-scale strategy, while in~\cite{andersson_water-fat_2018} Andersson {\em et al.} added Gaussian smoothing to improve robustness against noise. R-GOOSE uses a surface estimation problem to impose spatial smoothness combined with a multi-scale, non-iterative graph-cut algorithm~\cite{cui_rapid_2018}. This is similar to Stelter {\em et al.}, which introduced a hierarchical multi-resolution approach with multiple graph-cuts focused on water-fat-silicone separation for breast MRI~\cite{stelter_hierarchical_2022}. The systematic review by Daudé {\em et al.} compares many of these algorithms in fair and exhaustive tests using both {\em in silico} and {\em in vitro} phantoms~\cite{daude_comparative_2024}. These benchmarks reveal that, despite the substantial advances in algorithms for water-fat separation, there are still discrepancies between the recovered MRI-PDFF and \(R_2^*\) maps, showing high sensitivity to the choice of echo times, the fat spectrum and the strength of the inhomogeneities causing swap artifacts. This calls to take a step back from algorithm development to carefully analyze the signal model behind CSI. An understanding of this model can yield real insights on the conditions in which the recovery of the quantitative maps is actually possible and reliable.

This highlights both the challenges and the potential of CSI to quantify more than water and fat in the human body, the need to mitigate artifacts, and the importance of leveraging the spatial structure of the fieldmap to correctly recover the concentrations and \(R_2^*\). This leads us to analyze an abstraction of the water-fat separation problem in which the sample comprises many chemical species with their own radiation signal. In contrast to Magnetic Resonance Spectroscopic Imaging (MRSI)~\cite{posse_mr_2013} our model is closely related to CSI where the radiation pattern of the chemical species is known~\cite{mansfield_spatial_1984,trinh_vivo_2020}. The goal of analyzing this model is threefold. First, it allows us identify either obstructions or favorable conditions for the exact recovery of the concentrations, fieldmap, and \(R_2^*\). Second, it allows us to explain experimental results reported in the literature concerning the importance of the choice of echo times, and the smoothness of the fieldmap, allowing for the development of novel recovery techniques and efficient signal acquisition protocols. Third, our analysis yields result that may extend well beyond the water-fat separation problem, enabling the use of this technique to separate other quantities of physiological interest~\cite{shen_vivo_1999,stelter_hierarchical_2022} or other applications~\cite{poorman_multiecho_2019}. Neither our model nor its analysis assumes that the echo times are equispaced, which is a common assumption in methods based on ESPRIT~\cite{gudmundson_esprit-based_2012} and methods based on low-rank Hankel matrices~\cite{guo_fast_2017,guo_improved_2018,ying_vandermonde_2018}. Finally, the model that we analyze shares similarities with other widely used signal models, such as sum-of-exponential models~\cite{bouza_classification_2021} or modulated complex exponential models~\cite{yang_super-resolution_2016}, and thus our results may be readily applied in these models.

\subsection{Contributions}

In this work, perform a detailed analysis of a general separation problem of chemical species using CSI, focusing on finding suitable conditions under which recovery is possible. We determine the source of non-identifiability of the underlying quantities of interest and we propose a reconstruction method based on smooth non-convex optimization under convex constraints with recovery guarantees under suitable conditions. Our contributions are as follows.  

\begin{enumerate}[leftmargin=18pt, label=\roman*.]
    \item{{\bf Identifiability:} We determine the structure of the set of solutions to the inverse problem of characterizing the concentrations, the fieldmap and \(R_2^*\) from the measured signal under favorable conditions. This characterization depends only on the number of species assumed to be in the model, and on the echo times used for the acquisition.
    }
    \item{{\bf Oblique projections:} By leveraging {\em oblique} projections instead of {\em orthogonal} ones we are able to introduce a residual that is amenable both to analysis by means of the Wirtinger calculus and to an efficient computational implementation.
    }
    \item{{\bf Conditions for local convergence to the true parameters:} By leveraging the Wirtinger calculus we are able to show that a simple implementation of gradient descent with fixed stepsize converges to the true parameters in the noiseless case when the initial iterate is sufficiently close. We provide both a careful analysis and empirical evidence of how close this initial iterate should be. 
    }
    \item{{\bf Robustness to noise:} We provide a constrained variation of the method that is robust to noise. We also analyze the case of {\em model mismatch}, that is, when there are other chemical species in the sample that contribute to the measured signal.
    }
    \item{{\bf The imaging problem:} We propose a reconstruction method based on a smooth non-convex problem with convex constraints to address the imaging problem. This optimization problem can be solved with a proximal gradient descent method. We provide conditions under which our method recovers the true concentrations, the fieldmap and \(R_2^*\).
    }
\end{enumerate}

\subsection{Structure}

The manuscript is organized as follow. In Section~\ref{sec:signalModel} we introduce the signal model that constitutes the forward model for our analysis. This leads us to characterize in precise terms the solution set to the associated inverse problem under favorable  conditions. In Section~\ref{sec:theResidual} we leverage the notion of {\em oblique projections} to define a residual depending only on the fieldmap and \(R_2^*\) but not on the concentrations. This leads us to propose the minimization of the magnitude of this residual as the reconstruction procedure. In Section~\ref{sec:wirtingerCalculus} we briefly review the Wirtinger calculus and in Section~\ref{sec:localConvergenceGD} we prove that gradient descent with a fixed-step converges to the true parameter provided that the initial iterate is sufficiently close to the true fieldmap and \(R_2^*\). In Section~\ref{sec:stability} we propose a variation of our problem that is stable under noise and model misspecification. In Section~\ref{sec:imaging} we address the imaging problem, proposing to solve a {\em constrained} optimization problem instead of a {\em regularized} one. We establish connections between the constraints we propose and harmonic fieldmaps, and we prove that under suitable conditions the true parameters are the unique global minimizer to our problem. We defer the proofs of all of our main results to Section~\ref{apx:proofs}. Finally, Section~\ref{sec:experiments} presents the results of our numerical experiments.

\subsection{Preliminaries}

For \(n\in\N\) we define \(\bset{n} := \set{1,\ldots, n}\). Vectors are denoted in lowercase boldface and matrices in uppercase boldface. In \(\C^n\) the zero vector is denoted as \(\vnull_{n}\) and the standard complex inner product as \(\iprod{\cdot}{\cdot}\). The {\em support} of \(\vx\in\C^n\) is the set \(\supp(\vx)\) of indices \(k\in\bset{n}\) for which \(s_k\neq 0\). If \(a_1, \ldots, a_n \in \C\) we denote as \(\diag(a_1,\ldots, a_n)\) the \(n\times n\) matrix with entries \(a_1,\ldots, a_n\) along its diagonal. In \(\C\) we denote the upper half space as \(\Hp\).  We denote the least common multiple of \(p_1,\ldots, p_n\in \Z\) as \(\LCM(p_1,\ldots, p_n)\). If \(U\subset \R^n\) is open we say that \(f: U \to \R\) is {\em real analytic} if at every point in \(U\) there is a neighborhood on which it admits a power series expansion~\cite[Def.~2.2.1]{krantz_primer_2002}.

\section{The signal model}\label{sec:signalModel}

Consider a sample of \(n_s\) chemical species with concentrations \(c_1,\ldots, c_{n_s}\). If the \(\ell\)-th species resonates at a single frequency \(\Delta f_{\ell}\) relative to that of water, called the {\em chemical shift}, then the classical model for the signal generated by the sample is
\begin{equation}\label{eq:ideal:sp}
    s(t, \xi, c_1,\ldots, c_{n_s}) = e^{2\pi i \xi t} \sums_{\ell = 1}^{n_s} c_{\ell} e^{2\pi i \Delta f_\ell t}\quad\mbox{for}\quad \xi = \phi + i r_2^*
\end{equation}
where \(\phi\) is the {\em fieldmap} in Hz and \(r_2^* = R_2^*/2\pi\) is the normalized \(R_2^*\) in Hz. The model~\eqref{eq:ideal:sp} is called {\em single-peak} and can be improved by accounting for multiple resonance frequencies for each species. If \(\Delta f_1,\ldots, \Delta f_{n_p}\) comprise all the possible resonance frequencies for the species in the sample then the {\em multi-peak model} is
\begin{equation}\label{eq:ideal:mp}
    s(t, \xi, c_1,\ldots, c_{n_s}) = e^{2\pi i \xi t} \sums_{\ell = 1}^{n_s} c_{\ell} \vphi_{\ell}(t)
\end{equation}
where
\begin{equation}\label{eq:polyPhi}
    \vphi_{\ell}(t) = \sums_{p = 1}^{n_{p}} w_{\ell, p} e^{2\pi i \Delta f_p t}
\end{equation}
for weights \(w_{\ell, 1},\ldots, w_{\ell,n_p} \geq 0\) such that \(w_{\ell, 1} + \ldots + w_{\ell, n_p} = 1\). The weight \(w_{\ell, p}\) represents the fraction of energy that the \(\ell\)-th chemical species radiates at \(\Delta f_p\). 

Since~\eqref{eq:ideal:mp} generalizes~\eqref{eq:ideal:sp} we shall use the former as our signal model. We do not assume that \(\vphi_1,\ldots, \vphi_{n_s}\) have the form~\eqref{eq:polyPhi} but we do assume that their real and imaginary part are real analytic~\cite{krantz_primer_2002}. Furthermore, we assume that the concentrations can be complex, and that \(\xi \in \Hp\).

The signal~\eqref{eq:ideal:mp} is often sampled at \(n_e\) {\em echo times} \(t_1,\ldots, t_{n_e}\) which we assume are all positive, but may not be equispaced. This leads us to define the \(n_e\times n_s\) {\em model matrix} with entries
\begin{equation}\label{eq:modelMatrix}
    k\in\bset{n_e},\, \ell\in\bset{n_s}:\,\, \Phi_{k,\ell} = \vphi_{\ell}(t_k).
\end{equation}
and the {\em weighting matrix} \(\vW:\C\to \C^{n_e\times n_e}\)
\[
    \vW(\xi) = \diag(e^{2\pi i \xi t_1},\,\ldots,\, e^{2\pi i \xi t_{n_e}}).
\]
It will be useful to define the {\em signal matrix} \(\vM:\C\to \C^{n_e\times n_s}\) as 
\[
    \vM(\xi) = \vW(\xi)\vPhi(\xi)
\]
and the {\em signal map} \(s:\C\times \C^{n_s} \to \C^{n_e}\) as
\begin{equation}\label{eq:signalMap}
    s(\xi, \vc) = \vM(\xi)\vc. 
\end{equation}
Hence, if \(\xio\in \Hp\) represents the true parameter, and \(\vco\in \C^{n_s}\) the true concentrations, then the reconstruction problem consists on finding \(\xi\in \Hp\) and \(\vc\in \C^{n_s}\) such that
\[
    s(\xi, \vc) = \vso\quad\mbox{where}\quad \vso := s(\xio, \vco).
\]
Since the case of interest occurs when \(\vso \neq 0\) we shall assume from now on that \(\vco \neq 0\). Define the {\em solution set}
\[
    \sol(\vso) := \set{(\xi,\vc) \in \Hp\times \C^{n_s}:\,\, s(\xi,\vc) = \vso}.
\]
Ideally the pair \((\xio, \vco)\) would be the only element in this set. However, even under favorable conditions this is not the case. We first address two critical cases in which not even knowledge of either the true concentrations or parameters allows the unique determination of the other, to then identify favorable conditions under which we can characterize the solution.

\subsection{Some critical cases}\label{sec:criticalCases}

To determine how informative is the parameter when estimating the concentrations, observe that when \(n_e < n_s\) the nullspace of \(\vPhi\) is non-trivial. In this case we have that
\[
   \delta\vc\in \ker(\vPhi):\,\,  (\xi, \vc) \in \sol(\vso) \,\,\Rightarrow\,\, (\xi, \vc + \delta\vc) \in \sol(\vso).
\]
and, even if \(\xio\) is known, estimating \(\vco\) is hopeless. In contrast, when \(n_e \geq n_s\) and \(\vPhi\) is full-rank then \(\xio\) completely determines \(\vco\). Hence, from now on we assume that \(n_e \geq n_s\) and that \(\vPhi\) is full-rank. 

In contrast, if \(\vco\) is known, by letting \(\tvso = \vPhi\vco\) we observe that any \(\xi\) for which \((\xi,\vco)\) is a solution satisfies
\[
    \vW(\xi) \tvso = \vW(\xio)\tvso\,\,\Rightarrow\,\,\vW(\xi-\xio)\tvso = \vPhi\tvso.
\]
This implies that
\[
    k\in\supp(\tvso):\,\, e^{2\pi i (\xi - \xio) t_k} = 1.
\]
The number of equations is the size of the support of \(\tvso\). Intuitively, a larger number of equations imposes stronger constrains on the values of \(\xi -\xio\). However, although one expects to increase the number of equations by increasing the number of echos, there may be concentrations for which \(\tvso\) remains very sparse. For such concentrations, increasing the number of echo times does not yield additional information about the parameter.

To determine conditions under which \(\tvso\) is {\em never} too sparse, suppose that {\em every} \(n_s\times n_s\) submatrix of \(\vPhi\) is non-singular. In this case \(\supp(\tvso)^c\) can contain {\em at most} \(n_s -1\) elements, implying that \(\supp(\tvso)\) contains {\em at least} \(n_e - n_s + 1\) elements. In this case, there are always at least \(n_e - n_s + 1\) equations. Interestingly, this behavior is {\em generic} in a sense that we make precise. We defer the proof of the following result to Section~\ref{apx:modelMatrixSubmatrices}.

\begin{lemma}\label{lem:modelMatrixSubmatrices}
    Let \(T \in (0, \infty)\) and \(I = (0, T)\). Suppose that \(\vphi_1,\ldots,\vphi_{n_s}\) are linearly independent on \(I\) and that their real and imaginary parts are real-analytic on \(I\). There exists \(\mZ_{\Phi}\subset I^{n_e}\) with measure zero such that for any choice \((t_1,\ldots, t_{n_e})\notin \mZ_{\Phi}\) every \(n_s\times n_s\) submatrix \(\vPhi\) is non-singular.
\end{lemma}

Consequently, if a \(n_s\times n_s\) submatrix is singular, it suffices to perturb the echo times randomly over a sufficiently small interval to obtain a model matrix with this property. As a consequence of Lemma~\ref{lem:modelMatrixSubmatrices}, additional echo times contribute additional information about the parameter when \(\vco\) is known. We defer the proof of the following result to Section~\ref{apx:periodicitySetOfW}.

\begin{lemma}\label{lem:periodicitySetOfW}
    Let \(\vco\in\C^{n_s}\) and suppose that every \(n_s\times n_s\) submatrix of \(\vPhi\) is non-singular. If there is \(k,\ell\in \supp(\tvso)\) with \(k\neq \ell\) such that the quotient \(t_k / t_\ell\) is irrational, then \(\xio\) is uniquely characterized from \(\tvso\). Otherwise, there exists positive integers \(p, q\in \Z\) depending only on the echo times \(\set{t_k:\, k\in\supp(\tvso)}\) such that any
    \begin{equation}\label{eq:periodicitySetOfW}
        \xi \in \xio + \frac{1}{t_{\max}} \frac{q}{p}\Z
    \end{equation}
    where \(t_{\max} = \max(\set{t_k:\, k\in \supp(\tvso)})\) satisfies \((\xi, \vco) \in \sol(\vso)\).
\end{lemma}

If there are at least two {\em incommesurable} echo times with non-zero values then \(\vco\) completely determines \(\xio\). Otherwise, the proof of Lemma~\ref{lem:periodicitySetOfW} shows that we can write \(t_k/t_{\max} = p_k / q_k\) for \(k\in\supp(\tvso)\) and \(p_k,q_k\in \N\) whence
\begin{equation}\label{eq:integersForSolutionSet}
    p = \LCM(\set{p_k:\, k\in\supp(\tvso)})\,\,\mbox{and}\,\, q = \LCM(\set{p q_k / p_k:\, k\in\supp(\tvso)}).
\end{equation}
In practice, echo times have the form \(t_k = t_0 + k\Delta t\) for some \(t_0,\Delta t > 0\) and \(k\in\bset{n_e}\). In this case, if \(t_0 /\Delta t\) is irrational then there are two incommesurable echo times. In fact, if for some \(k\neq \ell\) the quotient \(t_k/t_\ell\) were rational, then there would be \(m_k,m_\ell \in \Z\) such that
\[
    m_\ell k - m_k \ell = \frac{t_0}{\Delta t}(m_k - m_\ell)
\]
contradicting the fact that \(t_0/\Delta t\) is irrational.

In practice, we do not know \(\xio\) nor \(\vco\). However, these results show the impact of the echo times even in this case, and they illustrate the best recovery guarantees that we can have. 

\subsection{Local identifiability}\label{sec:localInvertibility}

Although Lemma~\ref{lem:periodicitySetOfW} implies that we may not be able to recover \(\xio\), even when \(\vco\) is known beforehand, it does show that the set of \(\xi\) for which \((\xi,\vco)\in \sol(\vso)\) is {\em discrete} and thus \(\xio\) is the unique solution on a neighborhood around it. Hence, we say that \((\xio,\vco)\) is {\em locally identifiable} if there exists a neighborhood of \((\xio, \vco)\) containing no other element of \(\sol(\vso)\). Local identifability ensures that by restricting the possible values of \((\xi,\vc)\) we can still uniquely identify \((\xio,\vco)\) from \(\vso\). Our next result provides conditions under which {\em every} \((\xio,\vco)\) is locally identifiable. Once again the conditions are generic, and they depend on the number of echo times. We defer proof to Section~\ref{apx:localIdentifiability}.

\begin{theorem}\label{thm:localIdentifiability}
    Let \(T \in (0, \infty)\) and \(I = (0, T)\). Consider the collection \(\tilde{\vphi}_1,\ldots,\tilde{\vphi}_{2n_s}\) where \(\tilde{\vphi}_{\ell}(t) = \vphi_{\ell}(t)\) for \(\ell\in\set{1,\ldots,n_s}\) and \(\tilde{\vphi}_{\ell}(t) = t \vphi_{\ell}(t)\) for \(\ell\in \set{n_s,\ldots, 2n_s}\). Suppose that \(\tilde{\vphi}_1,\ldots, \tilde{\vphi}_{2 n_s}\) are linearly independent on \(I\) and that their real and imaginary parts are real-analytic on \(I\). If \(n_e \geq 2n_s\) then there exists \(\mZ_{\Phi} \subset (0, \infty)^{n_e}\) with empty interior and Lebesgue measure zero such that for any fixed choice of \((t_1,\ldots, t_{n_e})\notin Z_{\Phi}\) every \(n_s\times n_s\) submatrix of \(\vPhi\) is non-singular and {\em every} parameter \((\xio,\vco)\) is locally identifiable.
\end{theorem}

Therefore, when \(n_e \geq 2n_s\) then every \((\xio,\vco)\) is generically locally identifiable. We call these {\em favorable conditions}. In contrast, when \(n_s \leq n_e \leq 2n_s\) then some parameters are locally identifiable while other may not be. We provide sufficient conditions to determine when there is no local identifiability. Although its use is somewhat limited, it highlights the impact of the choice of echo times. We defer its proof to Section~\ref{apx:localIdentifiability}.

\begin{proposition}\label{prop:notLocalIdentifiable}
    Suppose that \(n_s \leq n_e \leq 2n_s\). If \((\xio,\vco)\) is not locally identifiable then there exists concentrations \(\vc\) such that
    \[
        \vT \vso = s(\xio,\vc) \,\,\mbox{where}\,\, \vT = \diag(t_1,\ldots, t_{n_e}).
    \]
\end{proposition}

\subsection{The structure of the solution set}\label{sec:solutionSet}

To characterize the solution set under favorable conditions, observe first that \(s(\xio + \eta, \vc) = s(\xio,\vco)\) if and only if
\[
    \vW(\xio) (\vW(\eta)\vPhi \vc - \vPhi \vco) = \vnull_{n_e}.
\]
This leads us to define \(\vDi : \C \times \C^{n_s}\times \C^{n_s} \to \C^{n_e \times 2n_s}\) as
\[
    \vDi(\eta) := \bmtx{\vW(\eta) \vPhi & \vPhi}.
\]
The nullspace of \(\vDi(\eta)\) yields substantial information about the solution set. We have that \((\xio+\eta,\vc) \in \sol(\vso)\) only if \(\ker(\vDi(\eta))\) is non-trivial. The converse is more subtle. Define 
\[
    Z_c(\eta) := \Lset{\vco \in \C:\, \exists\, \vc\in \C^{n_s}:\, \bmtx{\vc \\ -\vco}\in \ker(\vDi(\eta))}.
\]
This space represents {\em all} the \(\vco\) for which there exists some other concentration \(\vc\) for which \((\xio + \eta, \vc) \in \sol(\vso)\). Surprisingly, in special cases it is the case that \(\vc = \vco\). For example, when \(\vW(\eta) = \vI_{n_e}\) it is straightforward to see that \(\mZ_c(\eta) = \C^{n_s}\) and {\em every} \(\vco\) belongs to \(Z_c(\eta)\). However, from
\[
    \vnull_{n_e} = \bmtx{\vPhi & \vPhi}\bmtx{\vc \\ -\vco} = \vPhi(\vc - \vco)\,\,\Rightarrow\,\, \vc = \vco
\]
it follows that \(\vc = \vco\) even if \(\ker(\vDi(\eta))\) is non-trivial and \(\vco\in Z_c(\eta)\). The following proposition characterizes these cases. We defer the proof to Section~\ref{apx:solutionSet}.

\begin{proposition}\label{prop:recoverableConcentrationSet}
    Suppose that \(n_e \geq n_s\) and that every \(n_s\times n_s\) submatrix of \(\vPhi\) is non-singular. Let \(\vco\in \C^{n_s}\) be non-zero. Then 
    \[
        \bmtx{\vco \\ -\vco}\in\ker(\vDi(\eta))
    \]
    only if \(\imag(\eta) = 0\) and \(\eta\) belongs to a discrete set \(\Xi(\vco)\) depending only on the support of \(\tvso\) and the echo times. If there is \(k,\ell\in \supp(\tvso)\) with \(k\neq \ell\) such that the quotient \(t_k / t_\ell\) is irrational, then \(\Xi(\vco) = \set{0}\). Otherwise, there exists positive integers \(p, q\in \Z\) depending only on the echo times \(\set{t_k:\, k\in\supp(\tvso)}\) such that
    \[
        \Xi(\vco) = \frac{1}{t_{\max}} \frac{q}{p}\Z
    \]
    where \(t_{\max} = \max(\set{t_k:\, k\in\supp(\tvso)})\).
\end{proposition}

The integers \(p, q\) can be found in practice from~\eqref{eq:integersForSolutionSet}. This result implies that
\[
    \set{(\xio + \eta,\vco):\, \eta\in \Xi(\vco)} \subset \sol(\vso).
\]
Since Lemma~\ref{lem:periodicitySetOfW} establishes that {\em even when \(\vco\) is known} we may determine \(\xio\) up to a discrete set depending on the support of \(\tvso\) the best possible is equality in the above for all concentrations \(\vco\). This would imply that all possible solutions have {\em the same concentrations and the same \(r_2^*\)} and that the fieldmap \(\phi\) belongs to a discrete set. The following theorem establishes that this is generically the case when \(n_e \geq 2n_s\). We defer the proof to Section~\ref{apx:solutionSet}.

\begin{theorem}\label{thm:solutionSetStructure}
    Suppose that \(n_e \geq 2n_s\) and that every \(n_s\times n_s\) submatrix of \(\vPhi\) is non-singular. Then there exists a set \(\mZ_c\) with measure zero such that for any \(\vco\notin\mZ_c\) we have that
    \[
        \sol(\vso) = \set{(\xio + \eta, \vco):\,\, \eta \in \Xi(\vco)}.
    \]
\end{theorem}

\subsection{Swaps}

Our results in the previous section show that when \(n_e \geq 2n_s\) then there are no swaps and the concentrations are uniquely determined from \(\vso\). Therefore, swaps occur in the regime \(n_s < n_e < 2n_s\). Our next result shows that, depending on the structure of the model matrix, a form of generalized swaps will happen {\em generically}. We defer the proof to Section~\ref{apx:solutionSet}.

\begin{proposition}\label{prop:swaps}
    Suppose that \(n_s + 1\leq n_e \leq 2n_s\) and that every \(n_s\times n_s\) submatrix of \(\vPhi\) is non-singular. Then there exists a set \(\mZ_{sw}\) with empty interior and Lebesgue measure zero and an orthonormal basis \(\vu_1,\ldots, \vu_{n_s}\) of \(\C^{n_s}\) such that for any \(\vco\notin\mZ_{sw}\) there exists a discrete set of values for \(\eta\) and complex numbers \(e^{i\th_1},\ldots, e^{i\th_{n_s}}\) depending on \(\eta\) such that
    \[
        \vc = \sums_{\ell=1}^{n_s} e^{i\th_\ell}\iprod{\vu_\ell}{\vco}\vu_\ell
    \]
    satisfies \((\xio+\eta, \vc) \in \sol(\vso)\).
\end{proposition}

\section{The residual}\label{sec:theResidual}

The structure of the solution set suggests that under favorable conditions any solution to
\begin{equation}\label{eq:recoveryProblem}
    \vW(\xi)\vPhi \vc = \vso
\end{equation}
matches the true concentrations \(\vco\) while the parameter \(\xi\) may differ from the true parameter \(\xio\) only by a discrete amount. A popular approach to solve~\eqref{eq:recoveryProblem} is to use {\em variable projection (VarPro)}~\cite{golub_separable_2003,golub_differentiation_1973}. Since \(\vM\) is full-rank when \(\vPhi\) is, we can compute its Moore-Penrose pseudoinverse \(\vMp\) to estimate the concentrations from \(\vso\) as
\begin{equation}\label{eq:leastSquaresConcentrations}
    \cMP(\xi) = \vMp(\xi) \vso.
\end{equation}
Using this in~\eqref{eq:recoveryProblem} leads to
\begin{equation}\label{eq:leastSquaresSystem}
    \vso = \vM(\xi) \vMp(\xi) \vso
\end{equation}
where the variable is now the unknown parameter \(\xi\). The matrix \(\vM(\xi) \vMp(\xi)\) is the {\em orthogonal projector} onto \(\range(\vM(\xi))\). This {\em overdetermined} system of nonlinear equations can be solved to find an estimate \(\whxi\) of \(\xio\). In turn, this yields the estimate \(\whc = \cMP(\whxi)\vso\) for the concentrations. 

The system~\eqref{eq:leastSquaresSystem} is seldom solved directly in practice. Instead, the {\em residual norm squared}
\[
    \fo(\xi) = \frac{1}{2}\nrm{\vso - \vM(\xi) \vMp(\xi)\vso}^2_2
\]
is minimized. Unfortunately, this is computationally challenging due to the complex dependence that \(\vMp\) has on \(\xi\) and its conjugate \(\xi\adj\). For this reason, we propose an alternative approach based on {\em oblique projections} instead of orthogonal ones. From~\eqref{eq:recoveryProblem} we have that
\[
    \vW(-\xi)\vso = \vW(-\xi)\vM\vc = \vPhi\vc.
\]
Since \(\vPhi\) is full-rank, we can use its Moore-Penrose pseudoinverse \(\vPhip\) to obtain
\[
    \cRI(\xi) = \vPhip \vW(-\xi)\vso.
\]
Using this in~\eqref{eq:recoveryProblem} yields
\[
    \vso = \vW(\xi) \vPhi\vPhip\vW(-\xi) \vso.
\]
By solving this system of nonlinear equations we can obtain an estimate \(\whxi\) of \(\xio\) and the estimate \(\whc = \cRI(\whxi)\) for the concentrations. Instead, we proceed along the same lines as before, and observe that for every \(\xi\) the matrix
\[
    \vP(\xi) :=  \vW(\xi) \vPhi\vPhip\vW(-\xi)
\]
is an {\em oblique projector} as \(\vP(\xi)^2 = \vP(\xi)\) but \(\vP(\xi)\adj \neq \vP(\xi)\). This leads us to define the {\em residual matrix}
\begin{equation}\label{eq:residualMap}
    \vR(\xi) = \vW(\xi) \vPp \vW(-\xi)
\end{equation}
where \(\vPp = \vI_{n_e} - \vPhi\vPhip\) is the orthogonal projector onto \(\range(\vPhi)^{\perp}\). The residual matrix is also oblique projector. If \(\vR(\xi)\vso = 0\) then 
\[
    \vW(\xio - \xi) \vPhi\vco \in \range(\vPhi)\,\,\Rightarrow\,\, \vW(\xio - \xi) \vPhi\vco = \vPhi\vc
\]
for some \(\vc\) and thus \((\xi,\vc)\in\sol(\vso)\) as desired. The simple dependence of \(\vR\) on \(\xi\) not only allows us to show that it is complex differentiable, but also to compute its derivatives. Thus, we propose to minimize the {\em residual}
\[
    \fo(\xi) = \frac{1}{2} \nrm{\vR(\xi)\vso}_2^2.
\]
The simple dependence of the residual matrix on \(\xi\) makes the residual amenable to analysis using the Wirtinger calculus. 

\subsection{The Wirtinger calculus}\label{sec:wirtingerCalculus}

The Wirtinger calculus~\cite[Ch.~1, Sec.~4]{remmert_theory_1991} was introduced to analyze functions of a complex variable that are not complex differentiable. Instead of representing the {\em variable} \(\xi\) in terms of its real and imaginary part, the approach taken in Wirtinger calculus is to represent the {\em function} as depending on both the variable \(\xi\) and its complex conjugate \(\xi\adj\) while treating both as independent. From now on, we let \((\xi,\xi\adj)\) denote the Wirtinger variables and we let \(\nrmW{(\xi,\xi\adj)} = \sqrt{\xi\xi\adj}\). Lemma~\ref{lem:residualMatrixDerivatives} yields the representation
\[
    \fo(\xi, \xi\adj) = \frac{1}{2}\iprod{\vso}{\vR(\xi\adj)\vR(\xi)\vso}
\]
for the residual. It is apparent that \(\fo\) has derivatives of all orders on both \(\xi\) and \(\xi\adj\). Since \(\fo\) is real-valued, the identity
\begin{equation}\label{eq:wirtingerDerivativeRelations}
    \partial_{\xi\adj} \fo(\xi,\xi\adj) = \partial_{\xi} \fo(\xi,\xi\adj)\adj
\end{equation}
holds. Similar identities hold for higher order derivatives. The first-order Wirtinger derivatives are thus determined by
\begin{equation}\label{eq:residualDerivatives}
    \partial_{\xi} \fo(\xi,\xi\adj) = \frac{1}{2}\iprod{\vso}{\vR(\xi\adj)\vR'(\xi)\vso}.
\end{equation}
We denote as \(\nabla_W \fo\) the Wirtinger gradient and we represent its action on the variable \((\eta,\eta\adj)\) as
\[
    \nabla_W \fo(\xi,\xi\adj)(\eta,\eta\adj) = \real(\eta \iprod{\vso}{\vR(\xi\adj)\vR'(\xi)\vso}).
\]
It follows that the direction of maximum descent at \((\xi,\xi\adj)\) is \(-\nabla_W \fo(\xi,\xi\adj)\adj\). 

The second-order derivatives determine the curvature of the residual near \((\xi, \xi\adj)\). We denote the Wirtinger Hessian as \(\nabla_W^2 \fo(\xi,\xi\adj)\). The identity~\eqref{eq:wirtingerDerivativeRelations} implies that the Wirtinger Hessian is determined by
\begin{align*}
    \partial_{\xi\xi} \fo(\xi,\xi\adj)      &= \frac{1}{2}\iprod{\vso}{\vR'(\xi\adj)\vR'(\xi)\vso},\\
    \partial_{\xi\xi\adj} \fo(\xi,\xi\adj)  &= \frac{1}{2}\nrm{\vR(\xi)\vso}_2^2.
\end{align*}

\subsection{Gradient descent and exact local recovery}\label{sec:localConvergenceGD}

A simple strategy to solve
\begin{equation}\label{opt:minimizeResidual}
   \begin{aligned}
   & \underset{\xi}{\text{minimize}}
   & & \fo(\xi)
   \end{aligned}
\end{equation}
is to leverage Wirtinger calculus to use gradient descent with constant step size. This is also called {\em Wirtinger flow} in this context~\cite{candes_phase_2015,yuan_phase_2019}. A single iteration of gradient descent at \((\xi, \xi\adj)\) is represented as
\[
    (\xi,\xi\adj)^+ = (\xi,\xi\adj) - \alpha \nabla \fo(\xi, \xi\adj)\adj
\]
where \(\alpha > 0\) is the step size. Since \(\fo\) is non-convex, it is not clear whether the iterates will converge to a global minimum, or at all. Thus we provide conditions on the initial iterate and the step size than ensure that gradient descent converges to a global minimizer. 

The main idea behind our argument is as follows. Since \(\vR(\xio)\vso = 0\) it holds that
\begin{equation}\label{eq:hessianAtTrueParameter}
    (\eta,\eta\adj)\nabla^2_W \fo(\xio,\xio\adj)(\eta,\eta\adj) = |\eta|^2 \nrm{\vR'(\xio)\vso}_2^2.
\end{equation}
It follows from Lemma~\ref{lem:hessianAtTrueParameterPositive} that this is positive under general conditions. Therefore, the true parameter is a global strong minimizer for the residual. Since \(\fo\) has continuous second-order derivatives, its Hessian is positive definite in a neighborhood of \(\xio\) and, in fact, it is convex in a neighborhood of \(\xio\). Therefore, if the initial iterate is sufficiently close to the true parameter, gradient descent converges to the true parameter. The proof of the following theorem is deferred to Section~\ref{apx:localExactRecovery}. Recall that the Lambert \(W\) function is defined by the equation \(W(z) e^{W(z)} = z\) for \(z\in \C\). Here we only consider its main branch~\cite{mezo_lambert_2022}.

\begin{theorem}\label{thm:localExactRecovery}
    Let \(\tS = 4\pi (t_{n_e} - t_1)\) and \(\tne = 4\pi t_{n_e}\) and for \(\rho \in (0, 1)\) define
    \[
        \gxio^+(\rho) = \frac{1}{4}\frac{\nrm{\vR''(\xio)\vso}_2}{\tne^{5/2}}\left(-1 + \sqrt{1 + 8(1-\rho)\tne^2 \frac{\nrm{\vR'(\xio)\vso}_2^2}{\nrm{\vR''(\xio)\vso}_2^2}}\right)
    \]
    Let \(r_0 > 0\) be such that
    \begin{equation}\label{eq:radiusOfMonotonicityBound}
        r_0 \leq \frac{1}{\tS} W\left(\frac{\tS e^{-\tS \imag(\xio)}}{2\nrm{\vso}_2^2}\gxio^+(\rho)\right)
    \end{equation}
    where \(W\) is the Lambert \(W\) function. If \(\xi\in \Hp\) is such that \(|\xi - \xio| < r_0\) then gradient descent with step
    \[
        \alpha < \frac{\rho}{2 + \rho}
    \]
    converges to \(\xio\).
\end{theorem}

The bound in~\eqref{eq:radiusOfMonotonicityBound} is quite restrictive and in practice the radius can be much larger. We have provided it in this form for its simplicity. A more precise, but still restrictive, bound is given in~\eqref{eq:radiusOfMonotonicityBound:loose}. A better bound can be obtained by solving the implicit inequality in~\eqref{eq:radiusOfMonotonicityBound:tighter} in the supplementary material. In general, the radius can be much larger than these two bounds suggest, as evidenced in our numerical experiments. The main advantage of~\eqref{eq:radiusOfMonotonicityBound} is thus its conciseness and that it allows us to identify a figure of merit, namely, the quotient
\[
    \frac{\nrm{\vR'(\xio)\vso}_2}{\nrm{\vR''(\xio)\vso}_2} 
\]
which determines the local curvature of the residual near the true parameter.

\subsection{Stability and robustness}\label{sec:stability}

In practice, the signal \(\vso\) is typically corrupted by additive Gaussian noise. However, another source of signal corruption is the misspecification of the model, that is, when there are more species than the ones accounted for in the model. In this case, we may write
\[
    \vy = \vso + \vW(\xio) \vPhi_M \vc_M + \sigma \vz \quad\mbox{where}\quad \vz \sim N(\vnull_{n_e}, \vI_{n_e}).
\]
where \(\vPhi_M\) and \(\vc_M\) are the model matrix and concentrations for the unaccounted species. Instead of trying to solve~\eqref{opt:minimizeResidual} with \(\vy\) replacing \(\vso\) we propose the {\em full residual}
\begin{equation}\label{eq:fullResidual}
    f(\xi, \vs) = \frac{1}{2}\nrm{\vR(\xi)\vs}_2^2
\end{equation}
to then solve the constrained problem
\begin{equation}\label{opt:minimizeResidualNoisy}
   \begin{aligned}
   & \underset{\xi,\vs}{\text{minimize}}
   & & f(\xi,\vs)
   & \text{s.t.}
   & & \nrm{\vy - \vs}_2 \leq \delta
   \end{aligned}
\end{equation}
for some factor \(\delta > 0\). It can be selected by using the fact that
\[
    \ev\bset{\nrm{\vy - \vso}_2} \leq \nrm{\vW(\xio)\vPhi_M\vc_M}_2 + \sigma\sqrt{n_e}.
\]
To solve~\eqref{opt:minimizeResidualNoisy} we can leverage the fact that the full residual has Wirtinger derivatives of all orders, and that the constraint is convex to use projected gradient descent. If \(\vso\) is feasible for the constraint, then is is apparent that \((\xio, \vso)\) is an optimal solution to~\eqref{opt:minimizeResidualNoisy}. 

One of the advantages of using a constraint is that this does not change the curvature of the residual. In particular, the curvature on the variable \(\xi\) near the true parameters \((\xio,\vso)\) remains unaffected. A disadvantage of this approach is that it allows for multiple global minimizers. On one hand, this is due to the fact that the continuity of the signal map implies there exists \(r_\delta > 0\) such that \(\nrm{\vy - s(\xi,\vc)} \leq \delta\) whenever \(|\xi - \xio| \leq r_\delta\) and \(\nrm{\vc - \vco}_2 \leq r_0\). On the other, this is due to the structure of the full residual, as \(f(\xi\opt,\vs\opt) = 0\) implies that \(f(\xi\opt, \alpha\vs\opt) = 0\) for any \(\alpha\in\C\). Although we cannot mitigate the former issue, we can mitigate the second by solving the regularized problem    
\begin{equation}\label{opt:minimizeResidualRegularized}
   \begin{aligned}
   & \underset{\xi,\vs}{\text{minimize}}
   & & f(\xi,\vs) + \eps\nrm{\vs}_2^2
   & \text{s.t.}
   & & \nrm{\vy - \vs}_2 \leq \delta
   \end{aligned}
\end{equation}
for some small \(\eps > 0\). The effect of the regularization term is that for a global minimizer \((\xi\opt,\vs\opt)\) we must have that either \(\vs\opt = 0\) of \(\nrm{\vy - \vs\opt}_2 = \delta\). This has practical advantages when the noise dominates as the estimated signal will be zero.

\section{The imaging problem}\label{sec:imaging}

The {\em imaging problem} consists on recovering the the values of the parameter \(\xio\) and the concentrations \(\vco\) over a spatial region. Let \(V\) be the 2D or 3D region to be imaged. We let \(\xio:V\to \C\) be the true {\em parameter map} and \(\vco: V\to \C^{n_e}\) be the true {\em concentration map}. The signal map~\eqref{eq:signalMap} now becomes
\[
    s(\xio, \vco, v) = \vW(\xio(v)) \vPhi \vco(v)
\]
and the reconstruction problem becomes finding \((\xio,\vco)\) such that
\begin{equation}\label{eq:imagingProblem}
    v\in V:\,\, s(\xi,\vco, v) = \vso(v)
\end{equation}
where \(\vso(v) = s(\xio,\vco, v)\) for \(v\in V\). As formulated, the reconstruction problem decouples in the positions \(v\in V\). In fact, we may use the residual 
\begin{equation}\label{opt:imaging:nls}
   \begin{aligned}
   & \underset{\xi}{\text{minimize}}
   & & \sums_{v\in V}\, f(\xi(v), \vso(v)).
   \end{aligned}
\end{equation}
However, there is {\em a priori} information about the spatial structure of the true parameter \(\xio\) that we can leverage to improve the performance of this approach. For instance, the fieldmap \(\phio\) is often known to be smooth over large regions except where there are discontinuities or susceptibility effects. We now propose a concrete assumption on the prior structure of \(\phio\) that provably improves the reconstruction performance. Although our methods are readily applicable to problems in 3D, we present our arguments in 2D to simplify the exposition.

\subsection{Gradient bounds and exact local recovery}\label{sec:gradientBounds}

A typical assumption is that the fieldmap \(\phio\) is smooth. The smoothness is often determined by the magnitude of its gradient \(\nabla \phio\) of the fieldmap. For simplicity, we use forward differences
\[
    v\in V:\,\, \nabla \phio(v) = \bmtx{\phio(v + e_1) - \phio(v) \\ \phio(v + e_2) - \phio(v)}.
\]
A popular approach to promote a gradient of small magnitude is to regularize~\eqref{opt:imaging:nls} with a term of the form
\[
    h_p(\xi) = \frac{1}{2}\sums_{v\in V}\nrm{\nabla \real(\xi)}_2^p.
\]
When \(p = 2\) we obtain Tikhonov regularization, whereas when \(p = 1\) we recover TV regularization. Although this approach has been successful in practice, it has the disadvantage that it changes the curvature of the full residual. For this reason, we take a different approach. Observe that the choice of stencil to compute the gradient determines how its magnitude constrains the variations of the field at \(v\). It implicitly defines the neighborhood
\[
    \ngh(v) := \set{v}\cup \set{v'\in V:\, \mbox{\(v = v' \pm e_1\) or \(v = v' \pm e_2\)}},
\]
where \(e_1,e_2\) are the canonical basis vectors in \(\R^2\), at every point. Hence, the local variations of the fieldmap at any \(v\in V\) con be bounded as
\[
    |\phio(v') - \phio(v)| \leq \max\set{\nrm{\nabla\phio(v')}:\,\, v'\in\ngh(v)}.
\]
for \(v'\in \ngh(v)\). Therefore, instead of regularizing the gradient, given a bound \(\epsg:V\to \Rp\) we define the convex set
\begin{equation}\label{eq:gradientBoundConstraints}
    \Cg := \set{\xi:V\to \C:\,\, \nrm{\nabla \real(\xi)(v)}_2 \leq \epsg(v),\,\, v\in V}
\end{equation}
and we propose to solve 
\begin{equation}\label{opt:imaging:gradientBounds}
   \begin{aligned}
   & \underset{\xi}{\text{minimize}}
   & & \sums_{v\in V}\, f(\xi(v), \vso(v))
   & \text{s.t.}
   & & \xi \in \Cg.
   \end{aligned}
\end{equation}
In general, we assume that the truel fieldmap \(\phio\) is feasible for this problem. In this case, it is a global minimizer. Under suitable conditions, we can show that any other global minimizer cannot coincide at any point with \(\phio\). We defer the proof of the following result to Section~\ref{apx:localRecoveryImaging}.

\begin{theorem}\label{thm:localRecoveryImaging}
    Suppose that the conditions of Theorem~\ref{thm:solutionSetStructure} hold and that \(\vco\notin \mZ_c\). Let \(r_Z > 0\) be the distance between \(0\) and \(\Xi(\vco)\setminus\set{0}\). If \(\phi\opt\) is any other global minimizer for~\eqref{opt:imaging:gradientBounds} and there is \(v\in V\) such that \(\phi\opt(v) = \phio(v)\) then \(\phi\opt = \phio\). Therefore, global minimizers of~\eqref{opt:imaging:gradientBounds} cannot have the same value at at any point of \(V\).
\end{theorem}

This result supports our approach. As the constraints do not change the objective function, and thus the global minimizers, we can leverage the favorable structure of the residual to ensure a unique solution exists.

\subsection{Connection to harmonic fieldmaps}\label{sec:harmonicFieldmaps}

A common assumption of the fieldmap is that it is the superposition of a harmonic component and a perturbation due, for example, to susceptibility effects~\cite{zhou_background_2014,bao_whole_2019,bredies_perfect_2019}. Our constraints implicitly enforce this structure. In fact, if \(v\in \Vint\) then
\[
    \Delta \phio(v) = \nabla\adj \nabla \phio(v) = \phio(v + e_1) + \phio(v + e_2) - 4\phio(v) + \phio(v - e_1) + \phio(v - e_2). 
\]
For simplicity, assume that \(\epsg \equiv \epso\). If \(\phio\in \Cg\) then
\[
    |\Delta\phio(v)| \leq 4\epso.
\]
The constraints implicitly imply that an element of \(\Cg\) admits the decomposition
\[
    \phio = \phih + \psi
\]
where \(\Delta \phih = 0\) and \(\Delta\psi = \chi\) with \(|\chi| \leq 4\epso\). Therefore, our assumption is robust in the sense that it adapts to fieldmaps that may be small perturbations of harmonic maps. When \(\eps\) is not constant, then the same argument as before holds, but the magnitude of the source is location-dependent. 

\subsection{Stability and robustness}
\label{eq:noiseImaging}

When the signal \(\vso\) is corrupted by additive noise or by model misspecification of the model as discussed in Section~\ref{sec:stability} we propose to solve 
\[
   \begin{aligned}
   & \underset{\xi,\vs}{\text{minimize}}
   & & \sums_{v\in V}\, f(\xi(v), \vs(v))
   & \text{s.t.}
   & & \xi \in \Cg,\,\, v\in V:\, \nrm{\vs(v) - \vy(v)}_2 \leq \delta(v)
   \end{aligned}
\]
with projected gradient descent. The same arguments used in Section~\ref{sec:stability} show that if \(\vso\) if feasible for the constraint, then \((\xio,\vco)\) is a global minimizer for the above. 

\section{Experiments}\label{sec:experiments}

We present numerical experiments highlighting the practical consequences of our theoretical developments. We consider water, fat, and silicone as the chemical species potentially in the sample. For fat, we use the 6-peak model in~\cite{hamilton_vivo_2011} whereas for silicone we use the 1-peak model in~\cite{stelter_hierarchical_2022}. Our results can be reproduced with the code found in the GitHub repository~\href{https://github.com/csl-lab/CSITools}{{\tt csl-lab/CSITools}}.

\subsection{Structure of the solution set}\label{sec:experiments:solutionSet}

Both Lemma~\ref{lem:periodicitySetOfW} and Theorem~\ref{thm:solutionSetStructure} characterize the generic structure of the solution set when \(n_e \geq 2n_s\) and show its dependence on the choice of echo times. In this experiment we consider echo times of the form \(t_k = t_0 + \Delta t k\) where \(t_0 = 1.3\)ms and \(\Delta t = 1.05\)ms. 

When the concentrations \(\vco\) are known, Lemma~\ref{lem:periodicitySetOfW} relates the support of \(\tvso\) to the possible values that \(\xi\) can take so that \((\xi,\vco)\in \sol(\vso)\). Since all solutions differ only on their real part, we can evaluate the error \(\phi \mapsto \nrm{\vI_{n_e} - \vW(\phi)}_F\) to identify this set. To simplify the experiments, we assume that no component of \(\tvso\) is zero. Fig.~\ref{fig:wf:wset} shows the error when the sample contains only water and fat for \(4, 6, 7\) and \(8\) echo times. Observe that the {\em only} impact of increasing the number of echo times is in slightly increasing the magnitude of the error. However, the values at which the error is zero, and the oscillations of the error overall, do not change. This behavior persists when the sample contains water, fat and silicone. Fig.~\ref{fig:wfs:wset} shows the error for \(6, 7\) and \(8\) echo times. 

In contrast, Theorem~\ref{thm:solutionSetStructure} shows the structure of the solution set for \(n_e \geq 2n_s\) and generic concentrations \(\vco\). However, we can provide a complete analysis using the arguments in Section~\ref{apx:solutionSet} to approximate the set \(Z_{\Delta}\) of \(\eta\in\C\) such that \(\vDi(\eta)\) has non-trivial nullspace when \(n_s, n_e\) are not too large, and when the echo times are commensurable. In this case, we may write \(t_k/t_{n_e} = p_k/q_k\) for some \(p_k,q_k\in \N\) and make a change of variables \(z = e^{2\pi i \eta t_{n_e}/q}\) where \(q = \LCM(\set{q_1,\ldots, q_{n_e}})\). If we then define \(\tilde{p}_k = p_k q / q_k\) we can write
\[
  \vW(\eta) = \vW(z) = \diag(z^{\tilde{p}_1},\ldots, z^{\bar{p}_{n_e}}).
\]
As a consequence, if \(S\subset \set{1,\ldots, n_e}\) is a subset of \(2n_s\) elements and \(\vDi_S\) is the \(2n_s\times 2n_s\) submatrix formed by the rows with index in \(S\), the function \(\det(\vDi_{S})\) is a polynomial on \(z\). By finding the roots for all such polynomials, which is tractable in our case, we can determine the set \(Z_{\Delta}\). Remark that under the hypotheses of Theorem~\ref{thm:solutionSetRank} the roots lie in the unit circle or, equivalently, \(Z_{\Delta}\) is a subset of the reals. Fig.~\ref{fig:wf:sset} shows the smallest singular value of \(\vDi\) for a sample of water and fat for \(4, 6, 7\) and \(8\) echo times. The zeros of this function are precisely the elements in \(Z_{\Delta}\), and thus the values at which there may be either swaps, or where the concentrations may be recovered exactly. The green points are those at which the concentrations may be recovered exactly, and thus they correspond to \(\Xi(\vco)\) when \(\vco\notin \mZ_c\); the red points indicate that there may be swaps when \(\vco\in \mZ_c\). Increasing the number of echo times does not change the structure of this set. Fig.~\ref{fig:wfs:sset} shows similar results when water, fat and silicone are in the sample. The structure of the set is more complex and, in fact, swaps may be generated for different values of the fieldmap.

\subsection{Local recovery}\label{sec:experiments:localRecovery}

The fact that there is a neighborhood of \(\xio\) around which the residual \(\fo\) has positive curvature is critical to have local recovery conditions. We can compute this radius explicitly when the true parameters are known. We do this for a water, fat and silicone {\em in silico} phantom (see Section~\ref{sec:experiments:insilico} below) and we evaluate the bounds that we provided for this radius in Theorem~\ref{thm:localExactRecovery}. Fig.~\ref{fig:rmnt:voxels} shows
\[
    Q_{\xio,\vso}(r) = \min_{\xi\in \Hp:\, |\xi - \xio| = r}\,\frac{\nrm{\vR'(\xi)\vso}_2^2 - |\iprod{\vso}{\vR'(\xi)\vR''(\xi)\vso}|}{\nrm{\vR'(\xio)\vso}_2^2}
\]
for every voxel in the image. This quotient decreases by 50\% in a range between \(\approx 20\)Hz to \(\approx 55\)Hz and reaches zero by \(\approx 45Hz\) to more than 60Hz. This shows that the curvature of the residual can vary significantly between voxels. Fig.~\ref{fig:rmnt:image} shows the radius at which the value attains 50\% of its value, illustrating the spatial dependence on the signal at each voxel. Fig.~\ref{fig:rmnt:bnd} shows the estimates found by the bound in Theorem~\ref{thm:localExactRecovery}. As discussed earlier, it severely underestimates the radius, attaining a maximum value of \(\approx 0.35\)Hz. In contrast, in Fig.~\ref{fig:rmnt:bnd+} we show the estimate found using the improved bound in~\eqref{eq:radiusOfMonotonicityBound:tighter}. As can be seen, although it still underestimates the radius, it provides a bound of \(\approx 13\)Hz.

\begin{figure*}[!t]
    \centering
    \begin{subfigure}[t]{0.245\textwidth}
        \includegraphics[width=0.95\textwidth]{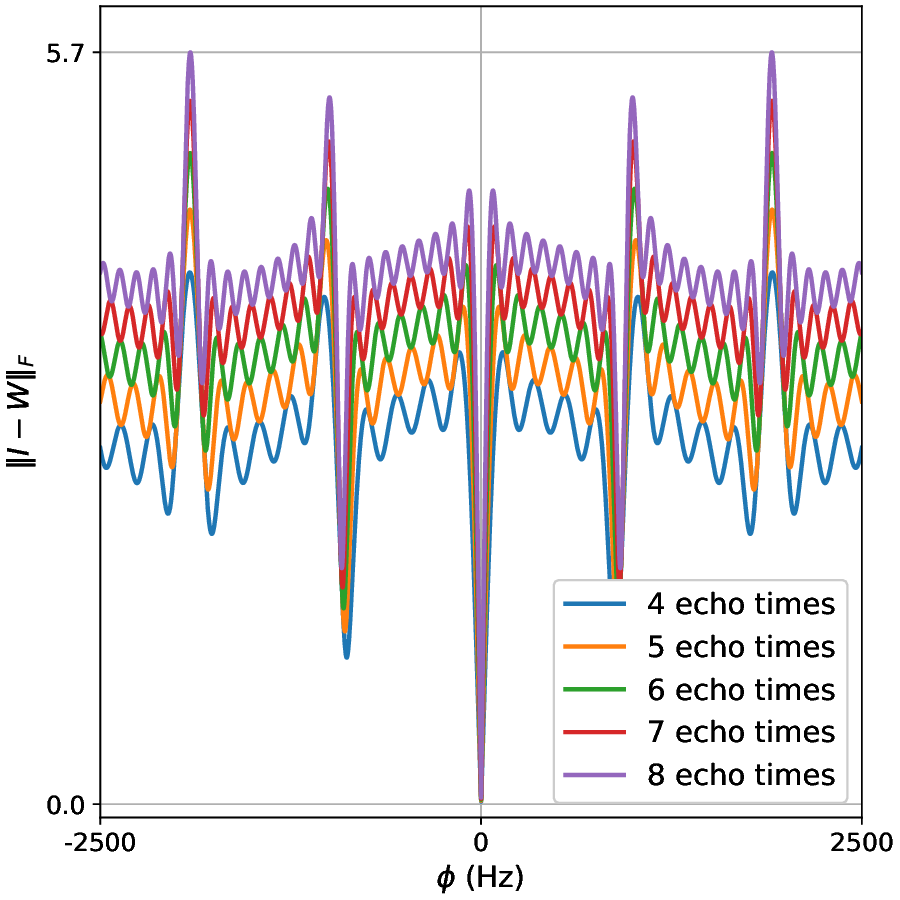}
        \caption{}
        \label{fig:wf:wset}
    \end{subfigure}%
    \begin{subfigure}[t]{0.245\textwidth}
        \includegraphics[width=0.95\textwidth]{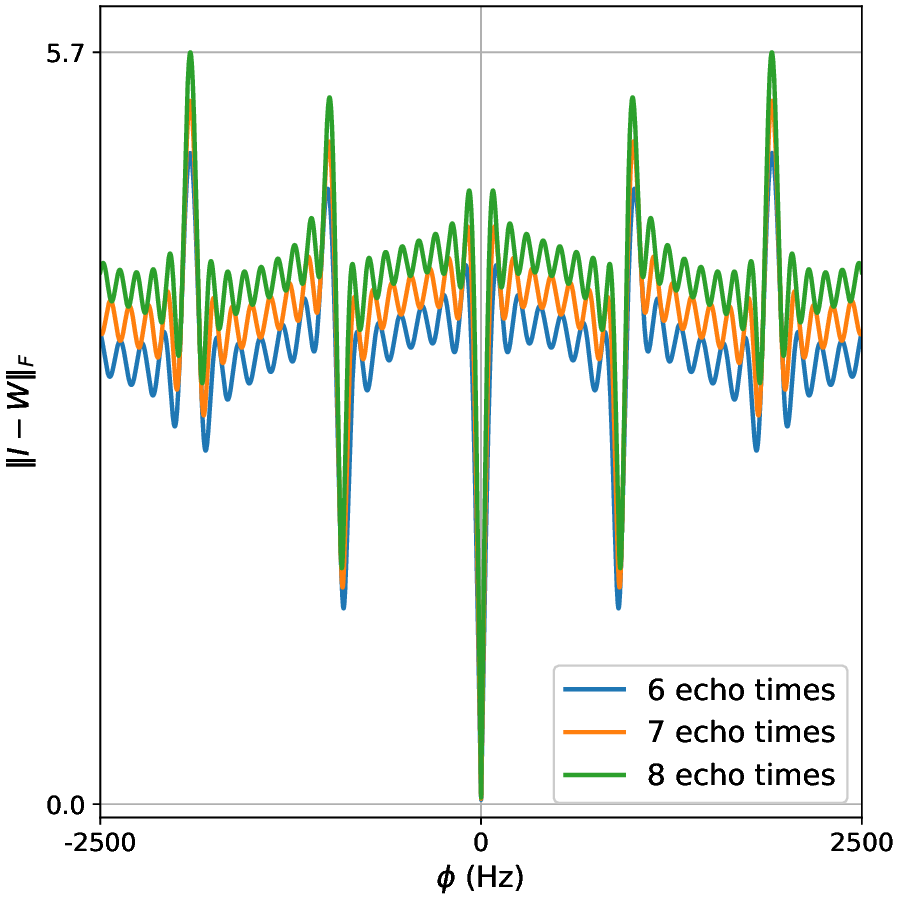}
        \caption{}
        \label{fig:wfs:wset}
    \end{subfigure}%
    \begin{subfigure}[t]{0.245\textwidth}
        \includegraphics[width=0.95\textwidth]{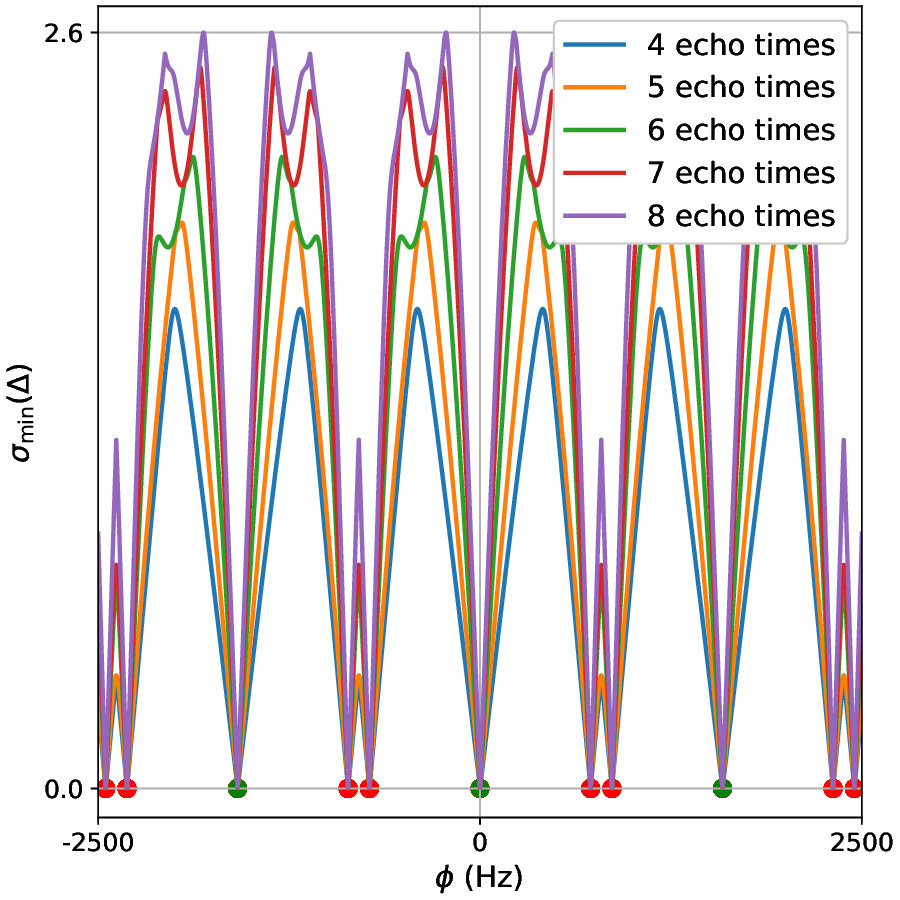}
        \caption{}
        \label{fig:wf:sset}
    \end{subfigure}%
    \begin{subfigure}[t]{0.245\textwidth}
        \includegraphics[width=0.95\textwidth]{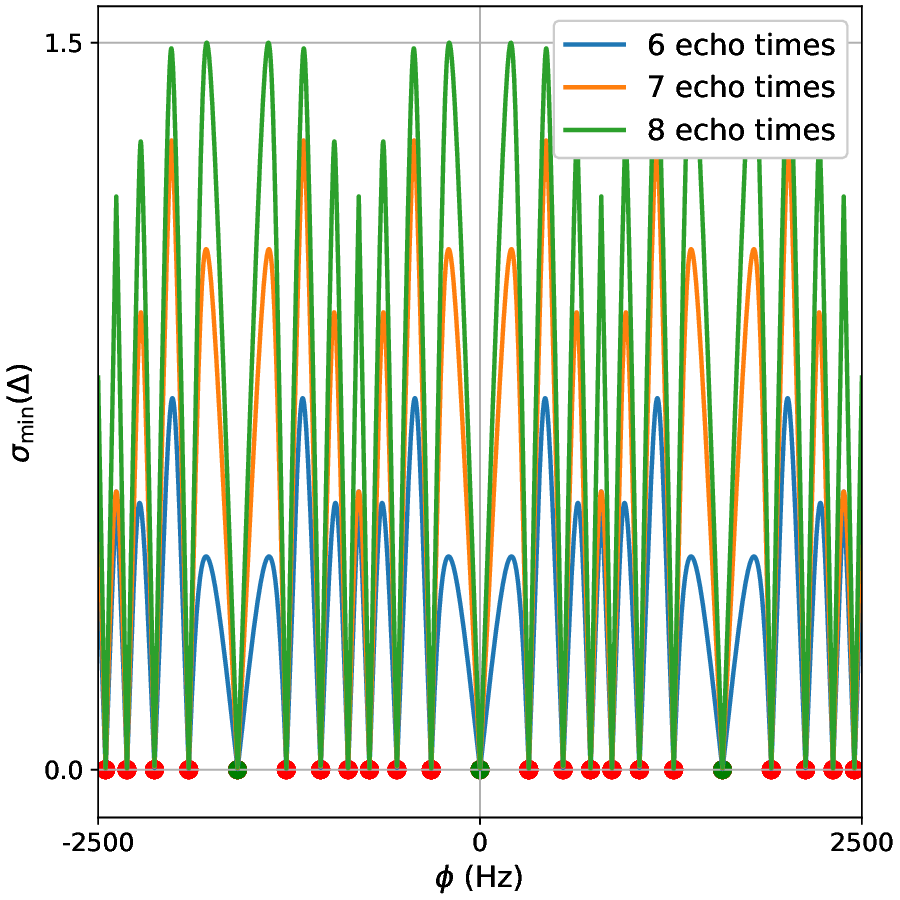}
        \caption{}
        \label{fig:wfs:sset}
    \end{subfigure}\\
    \vspace{3pt}
    \begin{subfigure}[t]{0.245\textwidth}
        \includegraphics[width=0.875\textwidth]{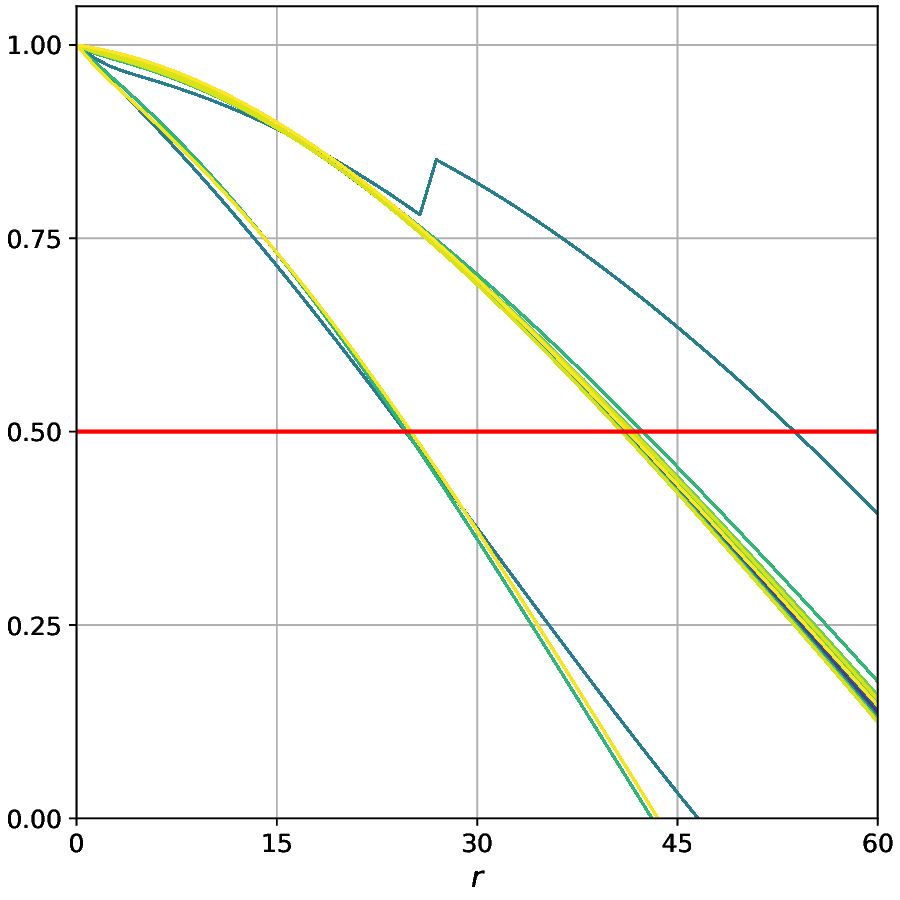}
        \caption{}
        \label{fig:rmnt:voxels}
    \end{subfigure}%
    \begin{subfigure}[t]{0.245\textwidth}
        \includegraphics[width=0.95\textwidth]{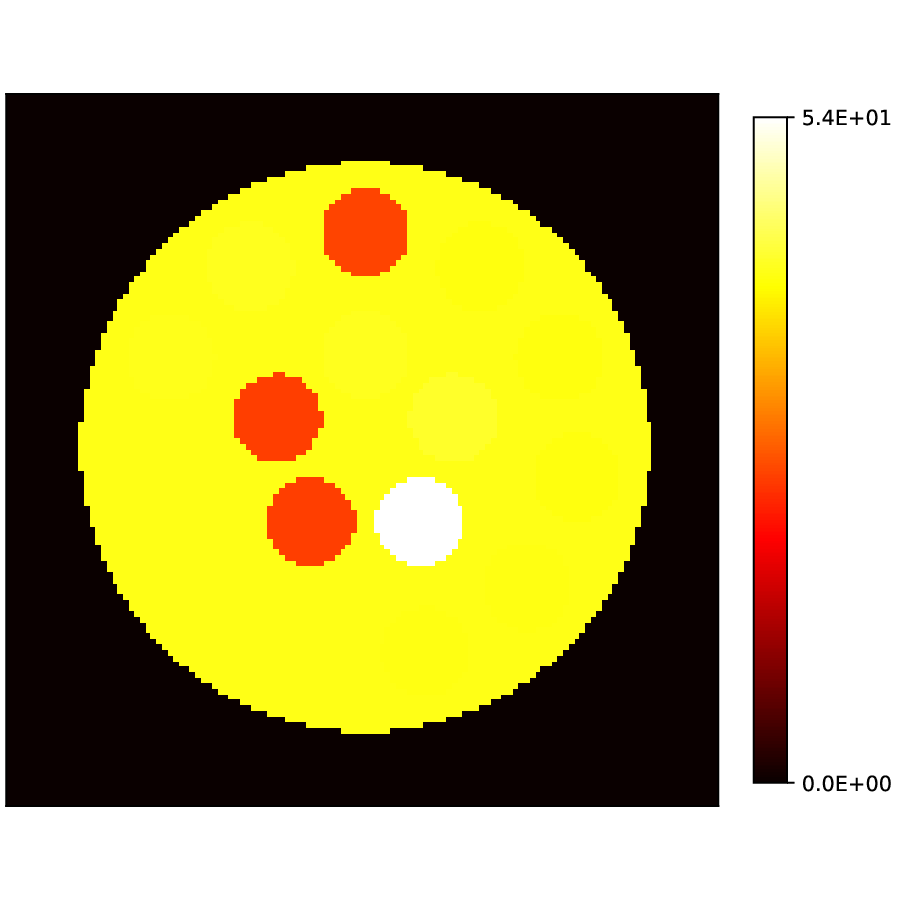}
        \caption{}
        \label{fig:rmnt:image}
    \end{subfigure}%
    \begin{subfigure}[t]{0.245\textwidth}
        \includegraphics[width=0.95\textwidth]{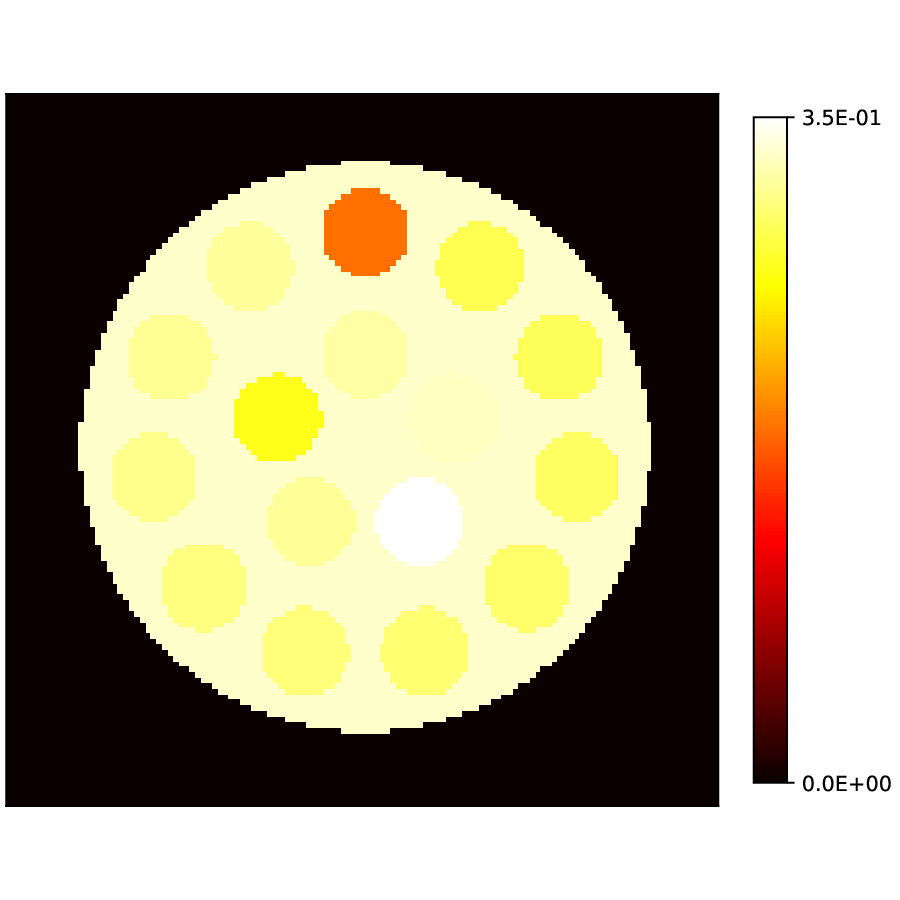}
        \caption{}
        \label{fig:rmnt:bnd}
    \end{subfigure}%
    \begin{subfigure}[t]{0.245\textwidth}
        \includegraphics[width=0.95\textwidth]{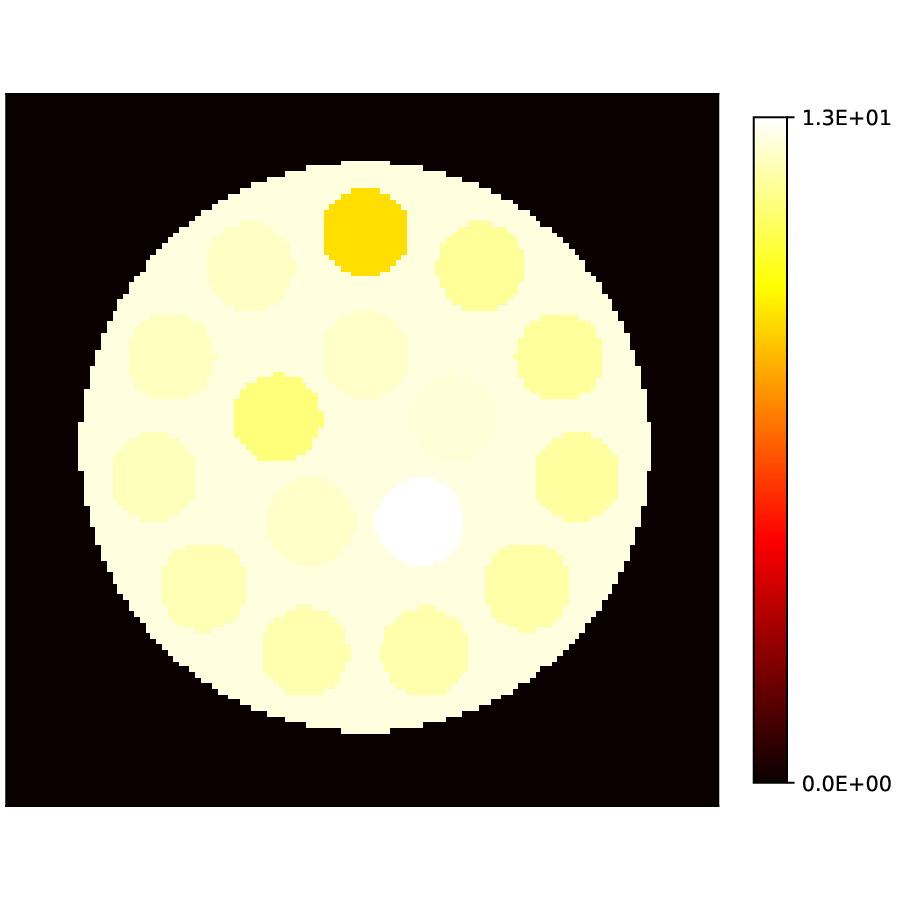}
        \caption{}
        \label{fig:rmnt:bnd+}
    \end{subfigure}%

    \caption{Structure of the solution set for a water and fat, and water, fat and silicone model. (a, b) Residual \(\nrm{\vI_{n_e} - \vW(\phi)}_F\) as a function of the fieldmap \(\phi\) for (a) water and fat and (b) water, fat and silicone. (c, d) Minimum singular value of \(\vDi(\phi)\) as a function of the fieldmap \(\phi\) for (c) water and fat and (d) water, fat and silicone. The green dots represents values for which the true concentrations may be recovered exactly, whereas the red dots represented the values for which some concentrations may be recovered with swaps. (e-h) Evaluation of the radius of positive curvature. (e) Fraction of reduction of the minimum eigenvalue of the Hessian at each voxel as a function of the distance for an {\em in silico} phantom of water, fat and silicone. (f) Radius at which the reduction is 50\%. (g) Estimate of the radius from the bound in Theorem~\ref{thm:localExactRecovery}. (h) Estimate of the radius from~\eqref{eq:radiusOfMonotonicityBound:tighter}.
    }
    \label{fig:exp:model}
\end{figure*}

\subsection{Experiments in an {\em in silico} phantom}\label{sec:experiments:insilico}

Our theoretical results show that generic concentrations \(\vco\) and \(r_2^*\) maps can be recovered exactly even when the fieldmap is not identifiable. To illustrate the impact of this fact, we perform a recovery experiment on a water (Fig.~\ref{fig:silico:wTrue}), fat (Fig.~\ref{fig:silico:fTrue}) and silicone (Fig.~\ref{fig:silico:sTrue}) {\em in silico} phantom. The concentrations are all real. The values for the fieldmap and \(r_2^*\) used to generate the signal are shown in Figs.~\ref{fig:silico:phTrue} and~\ref{fig:silico:rTrue}. The echo times have the form \(t_k = t_0 + \Delta t k\) where \(t_0 = 1.238\)ms and \(\Delta t = 0.986\)ms with \(k\in \bset{0,\ldots, 5}\).

We solve~\eqref{opt:imaging:gradientBounds} using projected gradient descent as initial iterate a vector \(\xi^{(0)}\) with all components equal to one. Forward finite-differences were used to compute the gradient. The bound on the norm of the gradient is \(30\)Hz at voxels with non-zero signal magnitude, and \(1\)kHz at voxels with zero signal magnitude. This avoids imposing artificial constraints at voxels with no signal. The step size used is \(\alpha = 10^3\) and the termination conditions 

In Figs.~\ref{fig:silico:wRec},~\ref{fig:silico:fRec} and~\ref{fig:silico:sRec} show the recovered concentrations of water, fat and silicone, and Fig.~\ref{fig:silico:rRec} shows the recovered \(r_2^*\). These recovered quantities are all qualitatively similar to their true values. In contrast, Fig.~\ref{fig:silico:phRec} shows the recovered fieldmap, which differs from its true value. By comparing the errors in the recovered concentrations, we see that they are within a reasonable accuracy except in regions with a large magnitude for the fieldmap gradient, indicating a bound that is too small (Figs.~\ref{fig:silico:wErr},~\ref{fig:silico:fErr} and~\ref{fig:silico:sErr}). A similar behavior is seen in the recovered \(r_2^*\) (Fig.~\ref{fig:silico:rErr}). The error for the recovered fieldmap tends to be larger outside the area of the phantom (Fig.~\ref{fig:silico:phErr}).

\begin{table}[h]
    \centering
    \small
    \begin{tabular}{|c||c|c|c|}\hline
                    & MSE & SNR (dB) & PSNR (dB) \\ \hline\hline
       water        & 2.00\(\times\)10\(^{-4}\) & 33.37 & 36.98 \\
       fat          & 2.06\(\times\)10\(^{-4}\) & 20.42 & 36.85 \\
       silicone     & 1.16\(\times\)10\(^{-3}\) & 10.13 & 29.35 \\
       fieldmap     & 4.71\(\times\)10\(^{+4}\) &  0.21 &  6.60 \\
       \(r_2^*\)    & 1.55\(\times\)10\(^{-2}\) & 26.73 & 38.59 \\\hline
    \end{tabular}\\
    \vspace{3pt}
    \caption{Error metrics for an {\em in silico} phantom.}
    \label{table:silico}
\end{table}

\begin{figure*}[!t]
    \centering
    \begin{subfigure}[t]{0.19\textwidth}
        \includegraphics[width=0.99\textwidth]{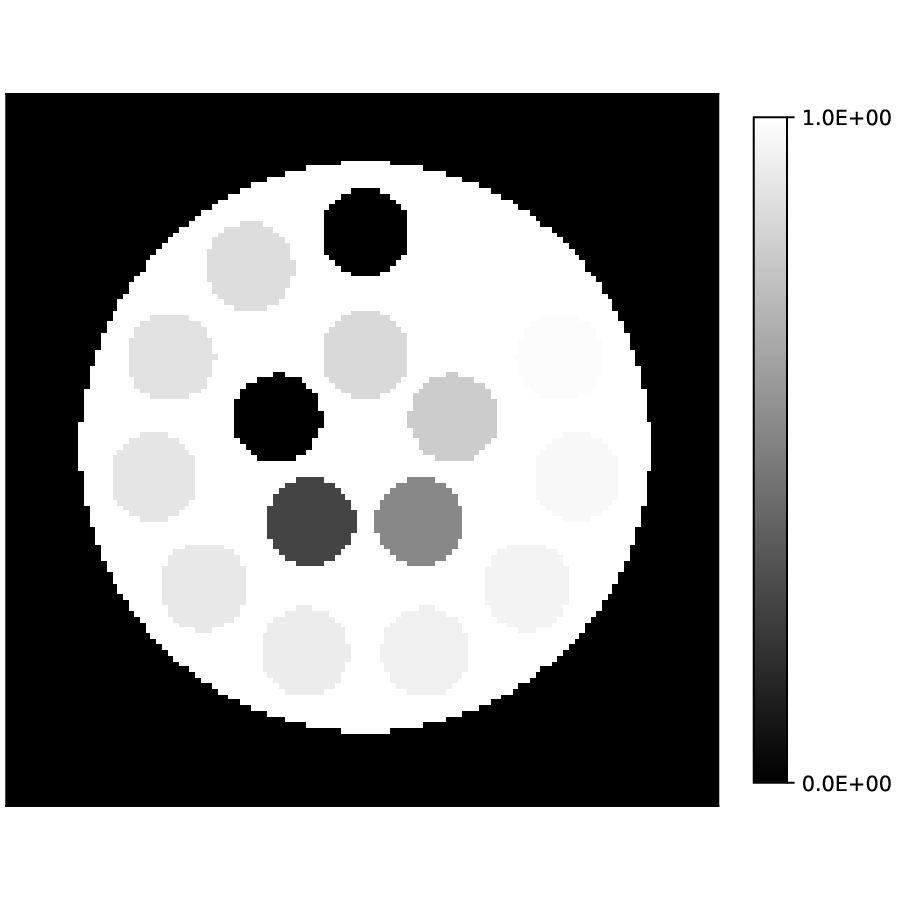}
        \caption{}
        \label{fig:silico:wTrue}
    \end{subfigure}%
    \begin{subfigure}[t]{0.19\textwidth}
        \includegraphics[width=0.99\textwidth]{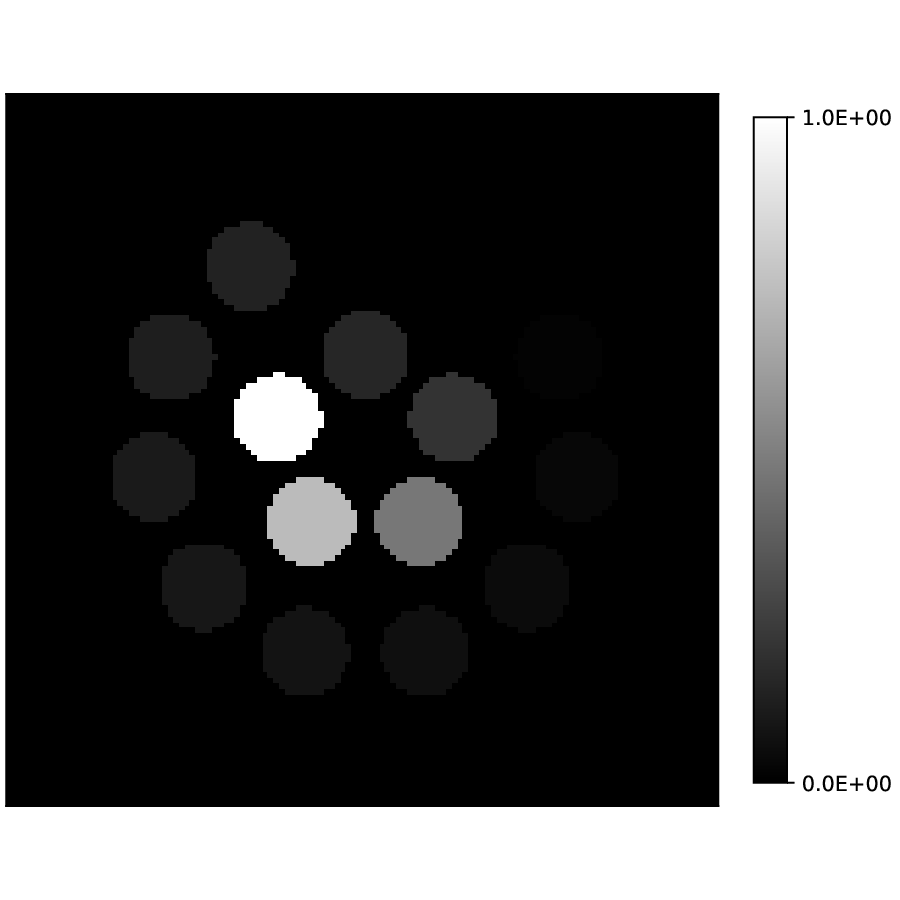}
        \caption{}
        \label{fig:silico:fTrue}
    \end{subfigure}%
    \begin{subfigure}[t]{0.19\textwidth}
        \includegraphics[width=0.99\textwidth]{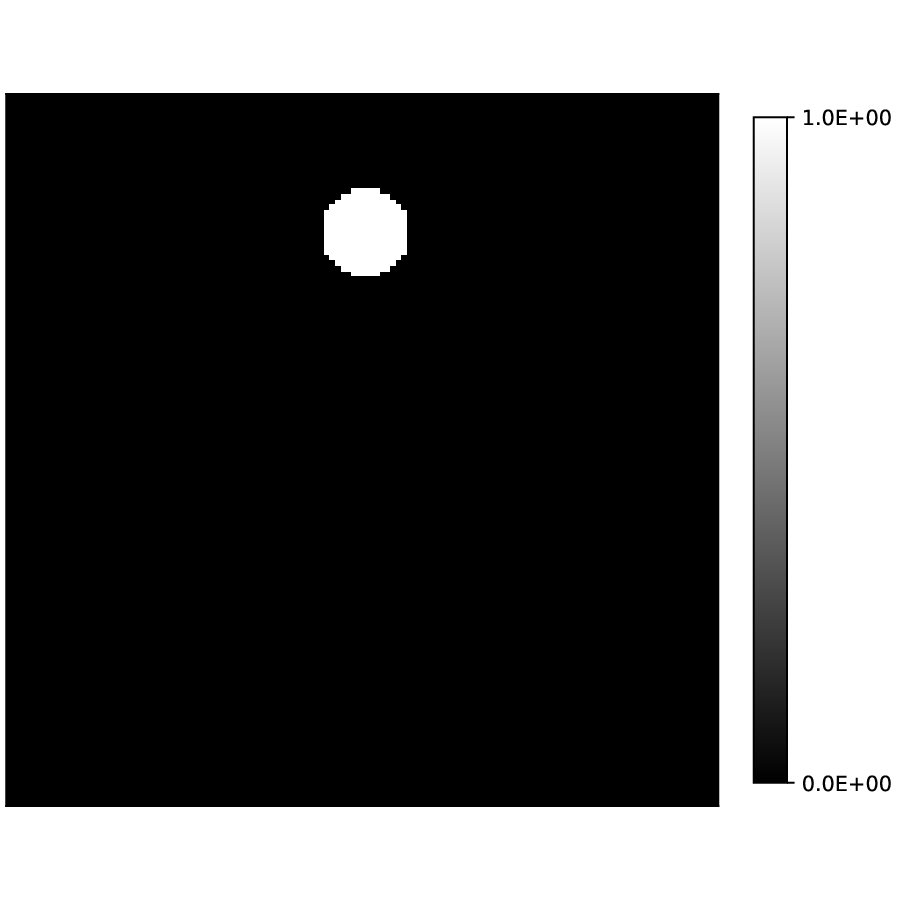}
        \caption{}
        \label{fig:silico:sTrue}
    \end{subfigure}%
    \begin{subfigure}[t]{0.19\textwidth}
        \includegraphics[width=0.99\textwidth]{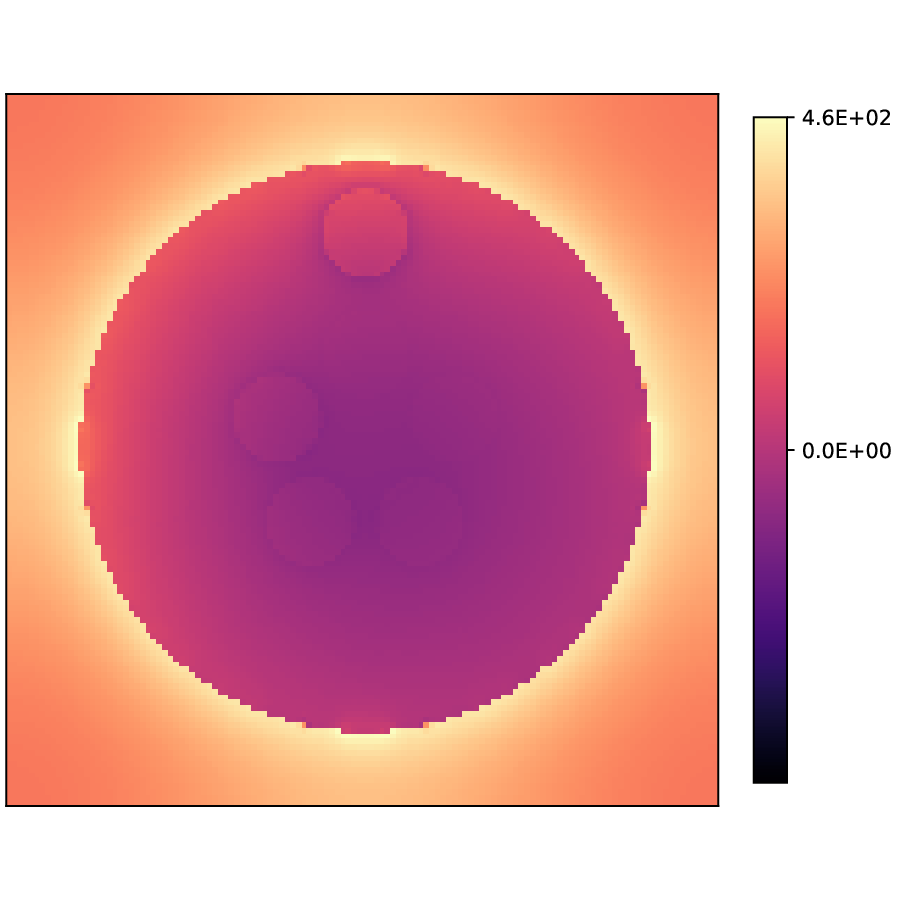}
        \caption{}
        \label{fig:silico:phTrue}
    \end{subfigure}%
    \begin{subfigure}[t]{0.19\textwidth}
        \includegraphics[width=0.99\textwidth]{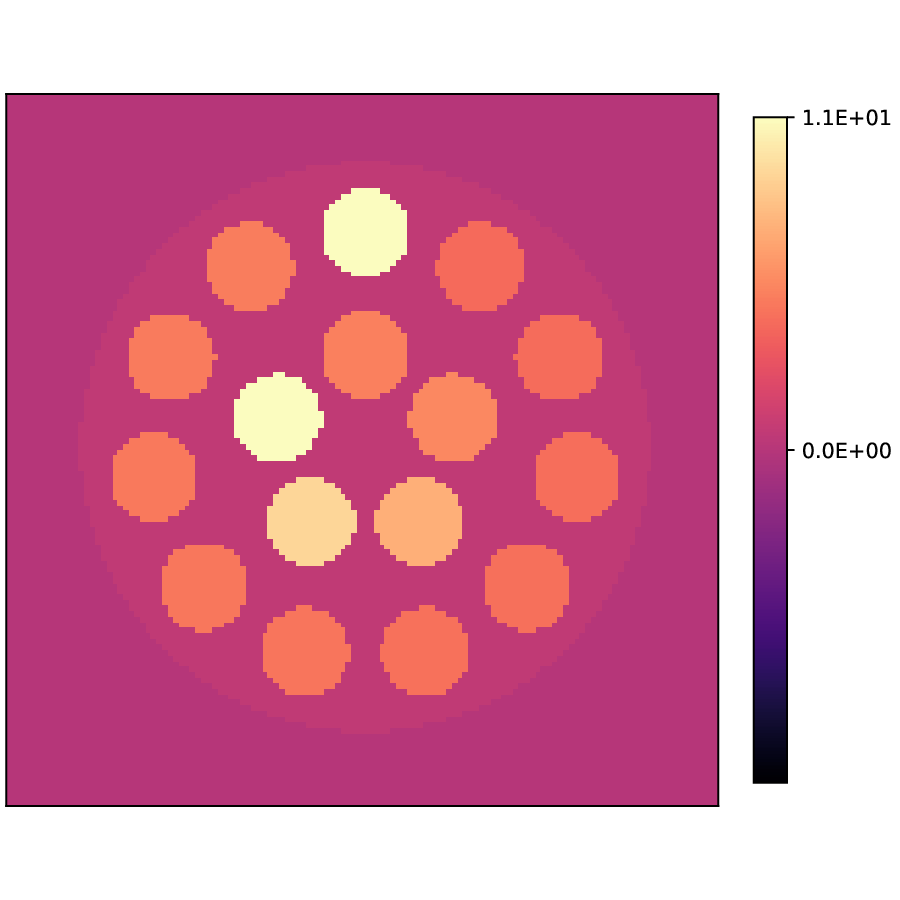}
        \caption{}
        \label{fig:silico:rTrue}
    \end{subfigure}\\
    \vspace{3pt}
    \begin{subfigure}[t]{0.19\textwidth}
        \includegraphics[width=0.99\textwidth]{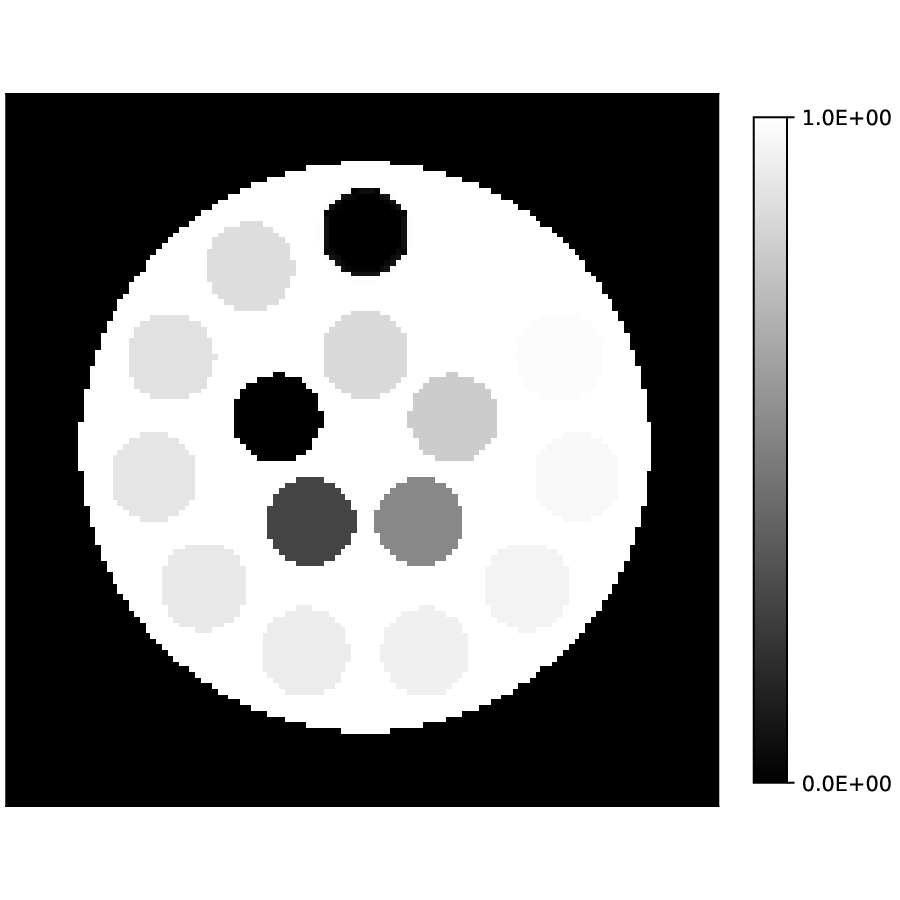}
        \caption{}
        \label{fig:silico:wRec}
    \end{subfigure}%
    \begin{subfigure}[t]{0.19\textwidth}
        \includegraphics[width=0.99\textwidth]{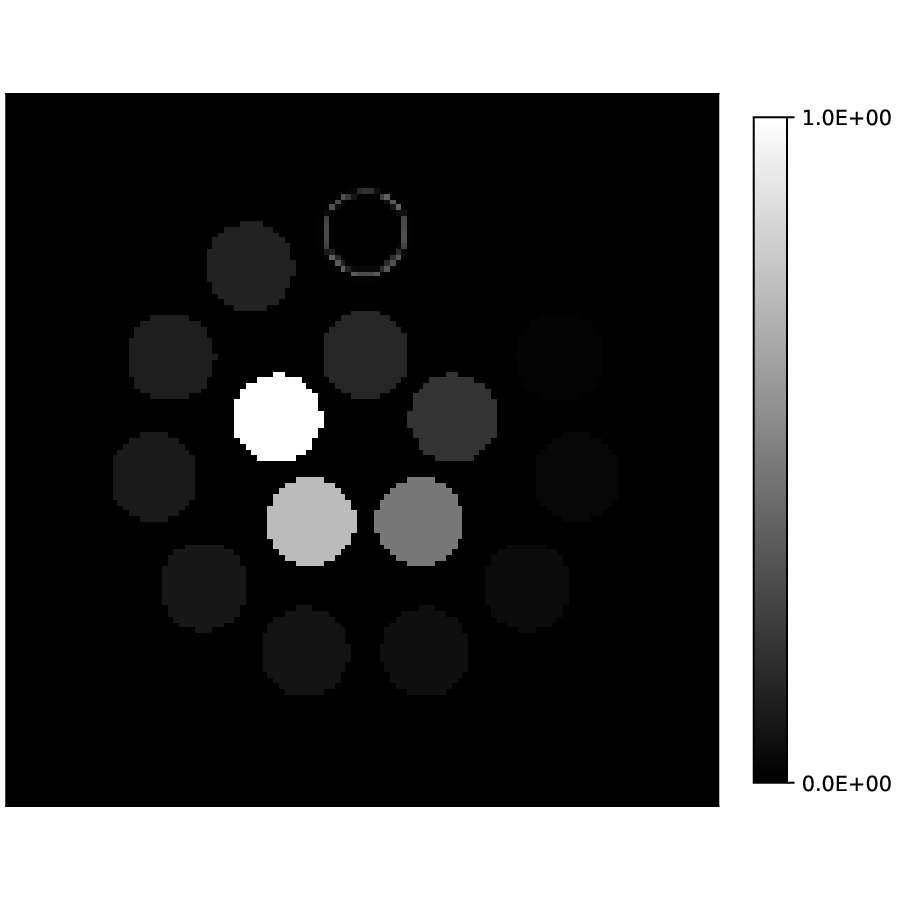}
        \caption{}
        \label{fig:silico:fRec}
    \end{subfigure}%
    \begin{subfigure}[t]{0.19\textwidth}
        \includegraphics[width=0.99\textwidth]{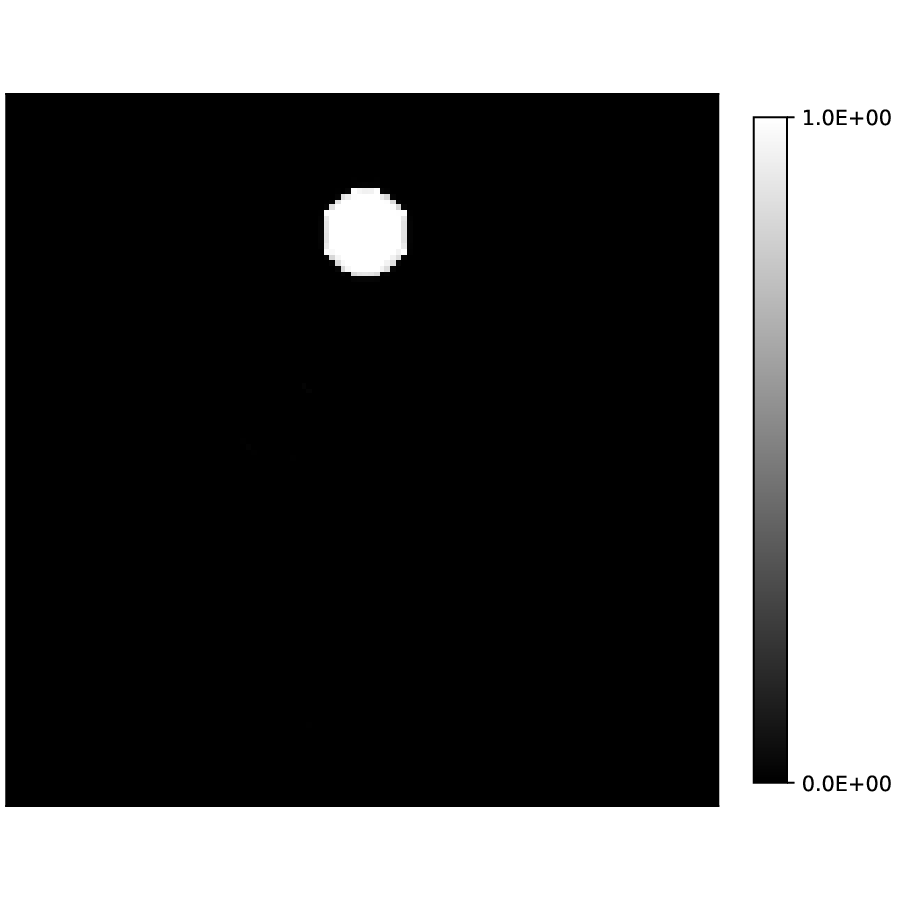}
        \caption{}
        \label{fig:silico:sRec}
    \end{subfigure}%
    \begin{subfigure}[t]{0.19\textwidth}
        \includegraphics[width=0.99\textwidth]{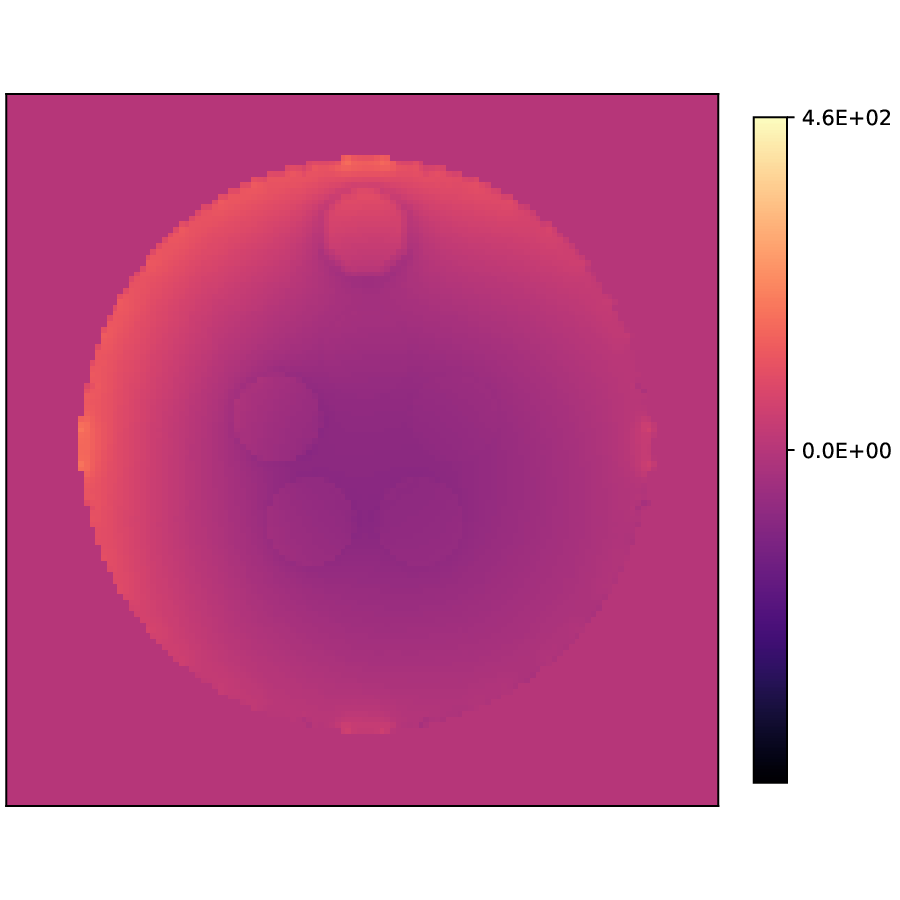}
        \caption{}
        \label{fig:silico:phRec}
    \end{subfigure}%
    \begin{subfigure}[t]{0.19\textwidth}
        \includegraphics[width=0.99\textwidth]{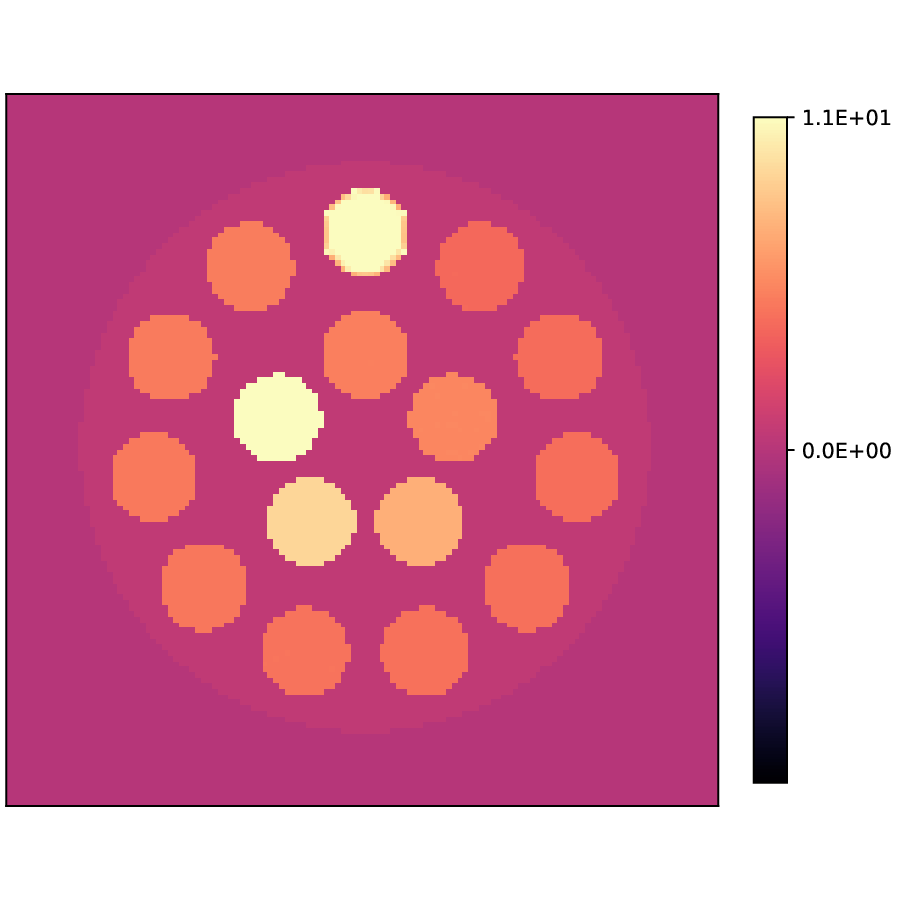}
        \caption{}
        \label{fig:silico:rRec}
    \end{subfigure}\\
    \begin{subfigure}[t]{0.19\textwidth}
        \includegraphics[width=0.99\textwidth]{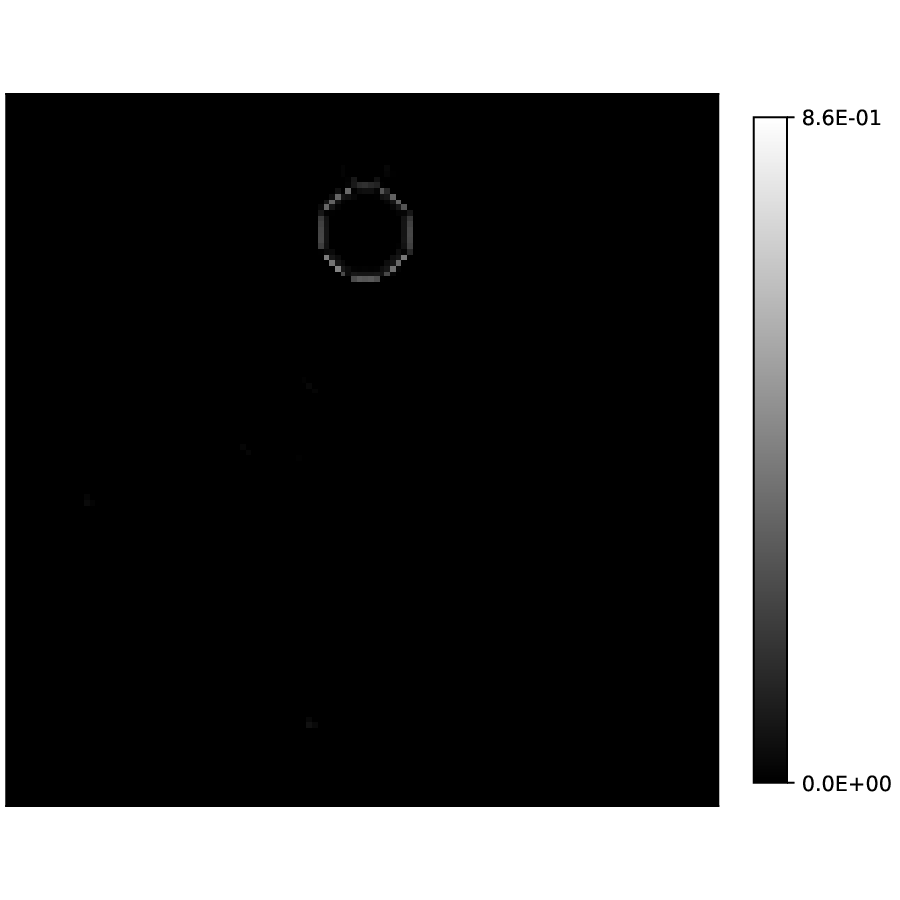}
        \caption{}
        \label{fig:silico:wErr}
    \end{subfigure}%
    \begin{subfigure}[t]{0.19\textwidth}
        \includegraphics[width=0.99\textwidth]{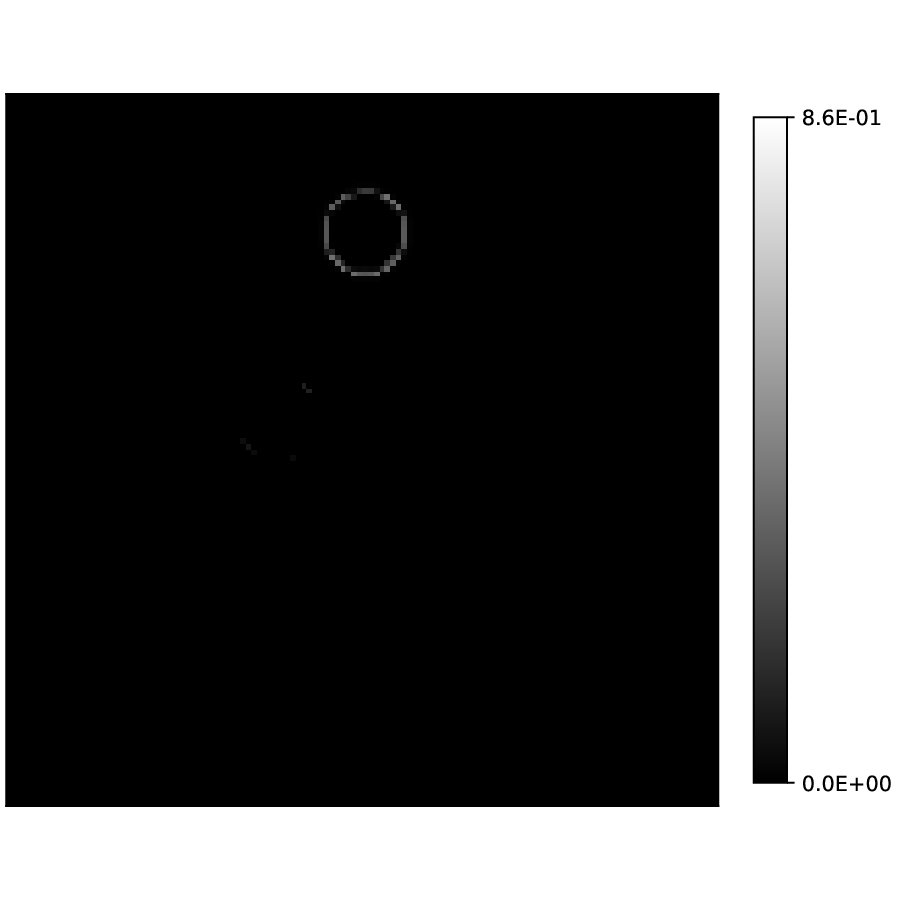}
        \caption{}
        \label{fig:silico:fErr}
    \end{subfigure}%
    \begin{subfigure}[t]{0.19\textwidth}
        \includegraphics[width=0.99\textwidth]{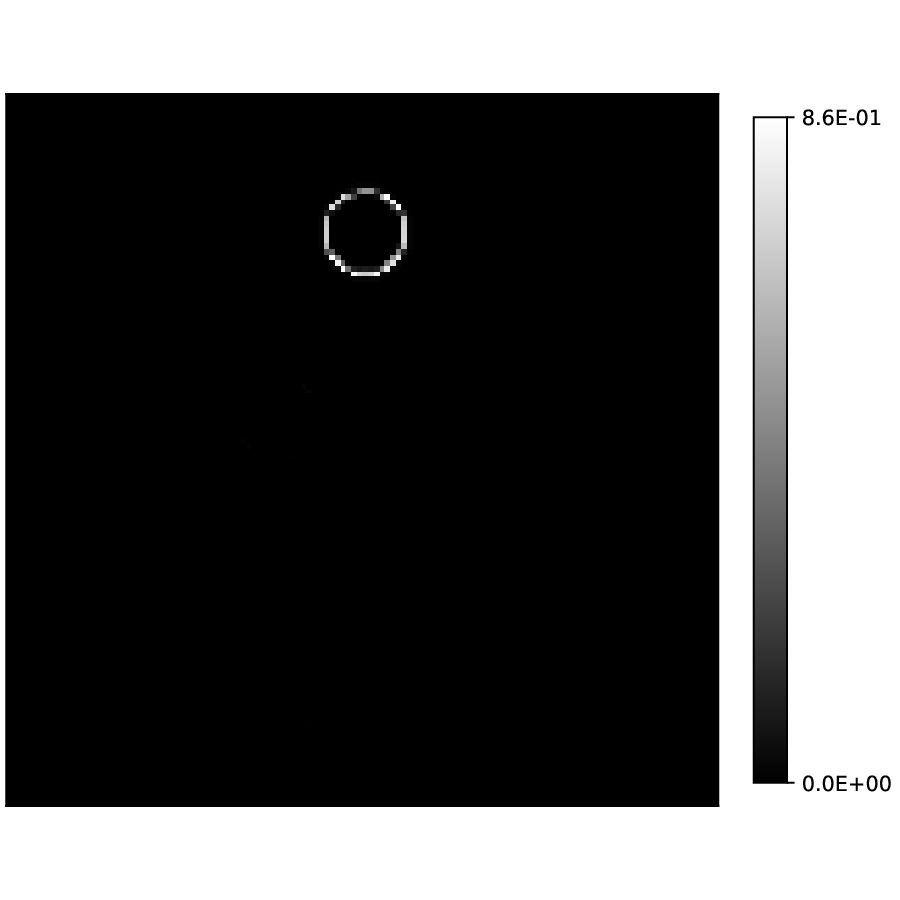}
        \caption{}
        \label{fig:silico:sErr}
    \end{subfigure}%
    \begin{subfigure}[t]{0.19\textwidth}
        \includegraphics[width=0.99\textwidth]{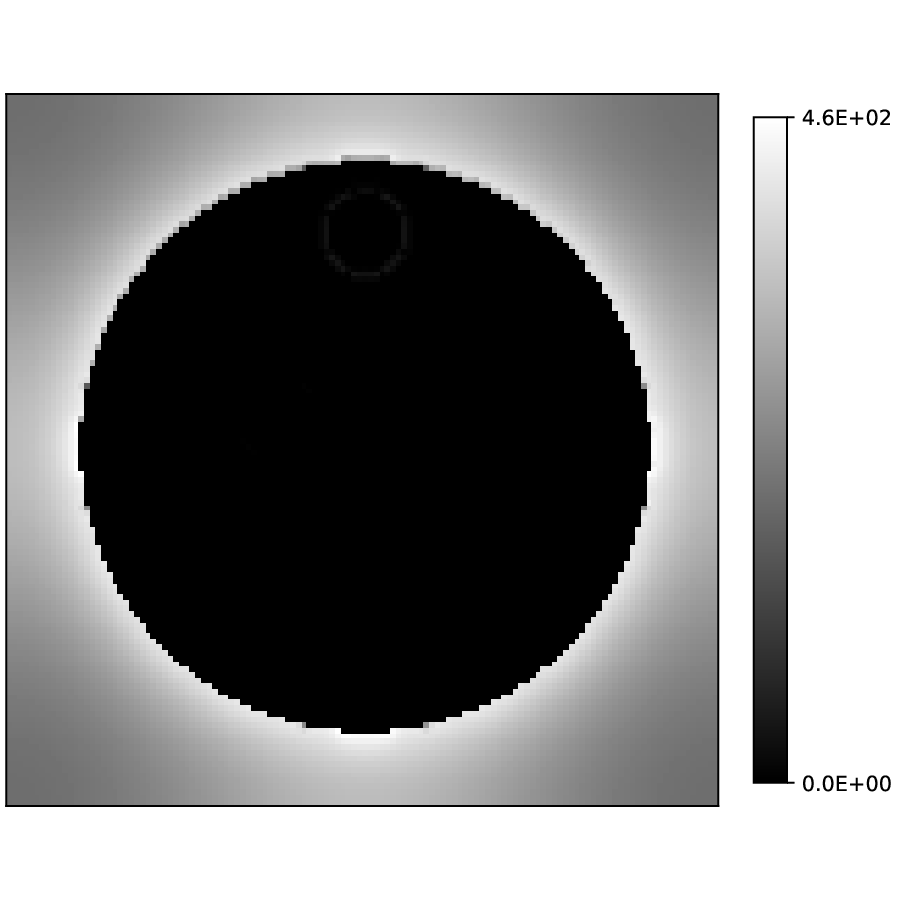}
        \caption{}
        \label{fig:silico:phErr}
    \end{subfigure}%
    \begin{subfigure}[t]{0.19\textwidth}
        \includegraphics[width=0.99\textwidth]{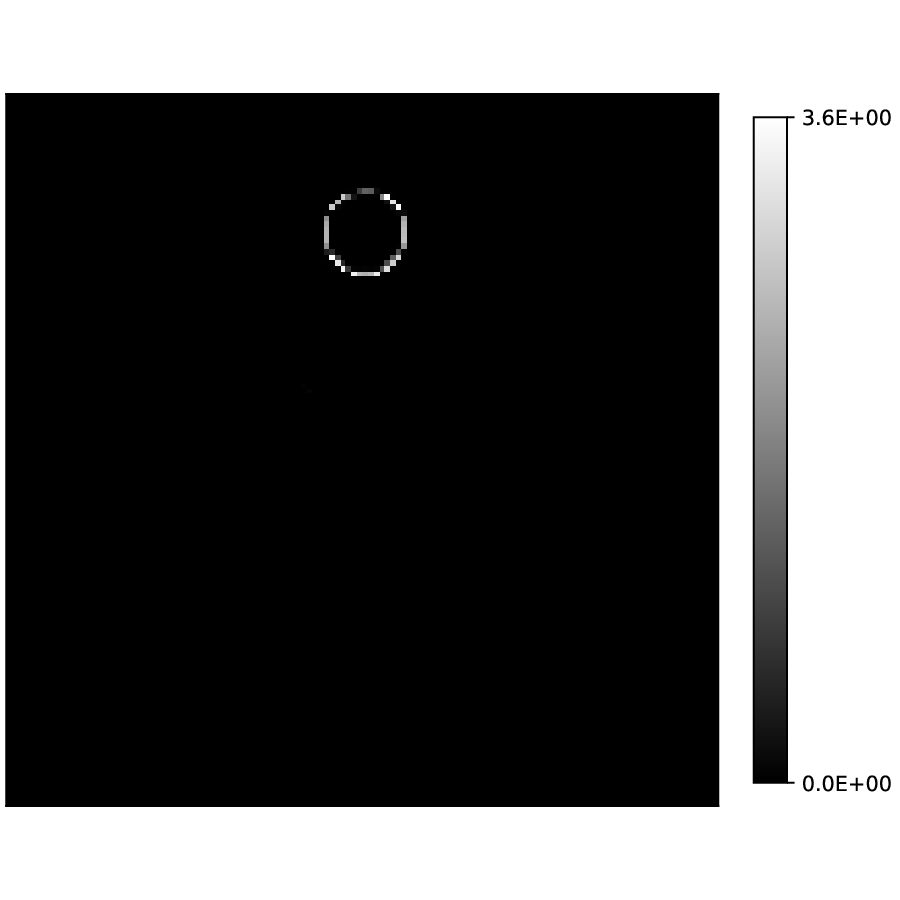}
        \caption{}
        \label{fig:silico:rErr}
    \end{subfigure}
    
    \caption{Recovery experiment in an {\em in silico} phantom. (a-e) True concentrations of (a) water, (b) fat and (c) silicone, (d) fieldmap and (e) \(r_2^*\). (f-j) Recovered concentrations of (f) water, (g) fat and (h) silicone, (i) fieldmap and (j) \(r_2^*\). (k-o) Differences between true and recovered concentrations of  (k) water, (l) fat and (m) silicone, and the (n) fieldmap and (o) \(r_2^*\) map.
    }
    \label{fig:exp:insilico}
\end{figure*}

\subsection{Experiments in an {\em in vitro} phantom}\label{sec:experiments:invivo}

To perform experiments using data from a {\em in vitro} water, fat and silicone phantom, we used a publicly available phantom introduced by~\cite{stelter_hierarchical_2022}. It consists of 15 vials with different mixtures of peanut oil and water, and a silicone structure placed the phantom. Following~\cite{stelter_hierarchical_2022} we use a 10-peak peanut oil model for fat, and a single peak at \(4.9\)ppm for silicone with a temperature correction of \(0.14\)ppm. The phantom was scanned at 3T with a gradient echo multi-echo acquisition with first echo at \(1.238\)ms and an echo spacing of \(0.986\)ms. 

To reconstruct the concentrations and the parameters, we used a bound on the magnitude of the gradient of \(20\)Hz at voxels with signal magnitude above 10\% of the maximum value of the signal accross the entire field of view.

Figs.~\ref{fig:invitro:wRec},~\ref{fig:invitro:fRec} and~\ref{fig:invitro:sRec} show the recovered concentrations of water, fat and silicone, and Figs.~\ref{fig:invitro:phRec} and ~\ref{fig:invitro:rRec} show the recovered fieldmap and \(r_2^*\). Fig.~\ref{fig:exp:invitro:error} shows the error of the recovered Proton Density Fat Fraction (PDFF) compared to the reference values obtained using Magnetic Resonance Spectroscopy (MRS-PDFF). Table~\ref{table:vitro} details the PDFF measurements in each vial.

\begin{table}[!h]
    \centering
    \small
    \begin{tabular}{|c||cc||}\hline
                    Vial & Recovered PDFF (\%) & MRS-PDFF (\%) \\ \hline\hline
        1       & 3.27 \(\pm\)  0.60 & 0 \\
        2       & 5.29 \(\pm\)  0.45 & 2.59 \\
        3       & 7.90 \(\pm\)  0.44 & 6.26 \\
        4       & 10.18 \(\pm\) 0.39 & 8.58 \\
        5       & 12.91 \(\pm\) 0.63 & 11.50 \\
        6       & 18.65 \(\pm\) 0.43 & 16.77 \\
        7       & 23.54 \(\pm\) 0.51 & 21.93 \\
        8       & 33.34 \(\pm\) 0.44 & 33.80 \\
        9       & 43.33 \(\pm\) 0.46 & 43.77 \\
        10      & 53.12 \(\pm\) 0.88 & 54.08 \\
        11      & 63.69 \(\pm\) 0.92 & 64.12 \\
        12      & 73.17 \(\pm\) 1.18 & 73.83 \\
        13      & 81.84 \(\pm\) 1.32 & 82.36 \\
        14      & 89.86 \(\pm\) 1.66 & 91.17 \\
        15      & 98.20 \(\pm\) 0.67 & 100.00 \\ \hline
    \end{tabular}\\
    \vspace{3pt}
    \caption{Error metrics for the {\em in vitro} phantom.}
    \label{table:vitro}
\end{table}

\begin{figure*}[!t]
    \centering
    \begin{subfigure}[t]{0.32\textwidth}
        \includegraphics[width=0.99\textwidth]{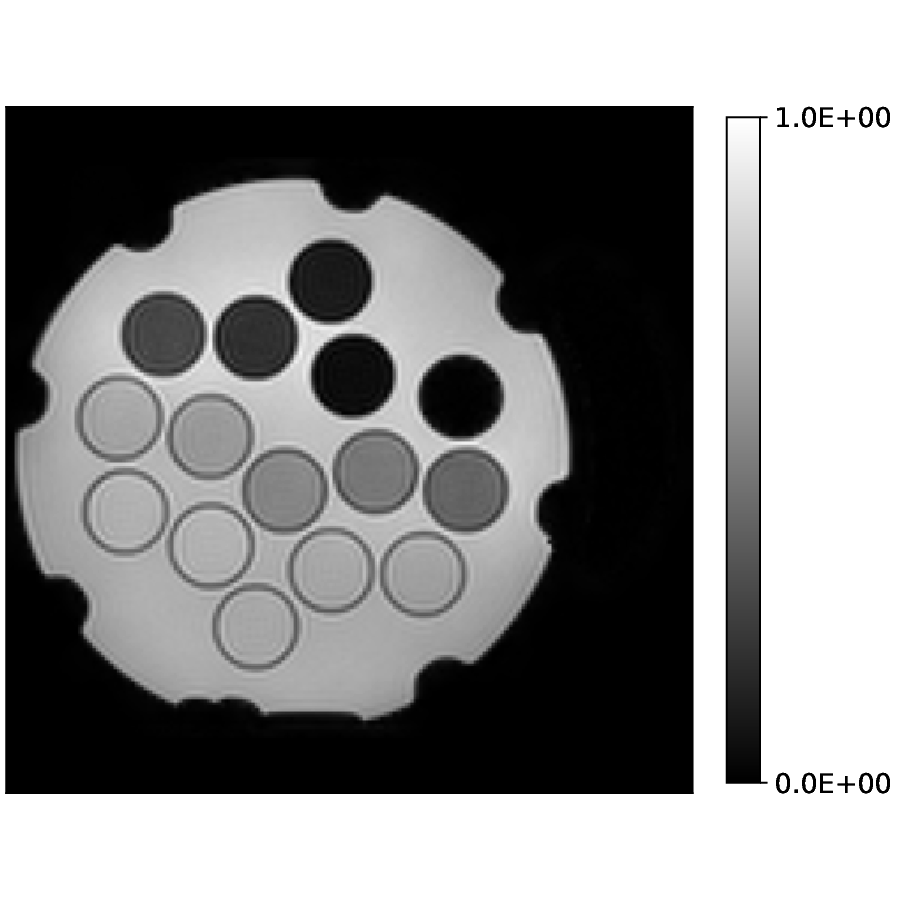}
        \caption{}
        \label{fig:invitro:wRec}
    \end{subfigure}%
    \begin{subfigure}[t]{0.32\textwidth}
        \includegraphics[width=0.99\textwidth]{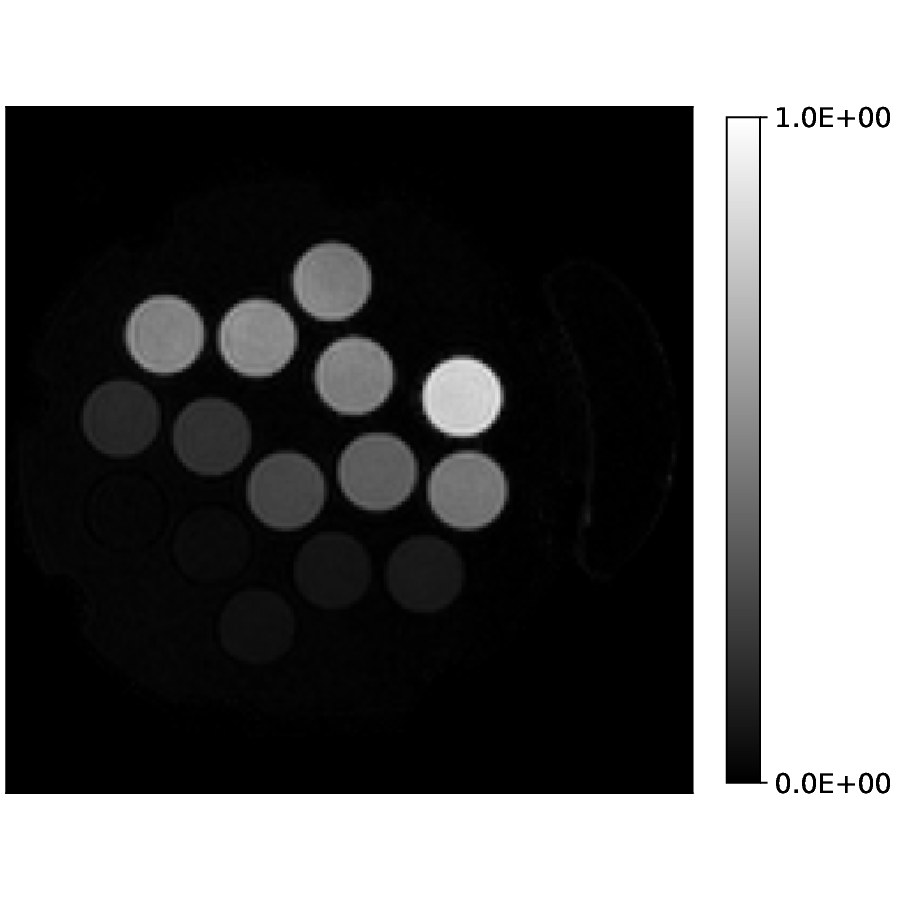}
        \caption{}
        \label{fig:invitro:fRec}
    \end{subfigure}%
    \begin{subfigure}[t]{0.32\textwidth}
        \includegraphics[width=0.99\textwidth]{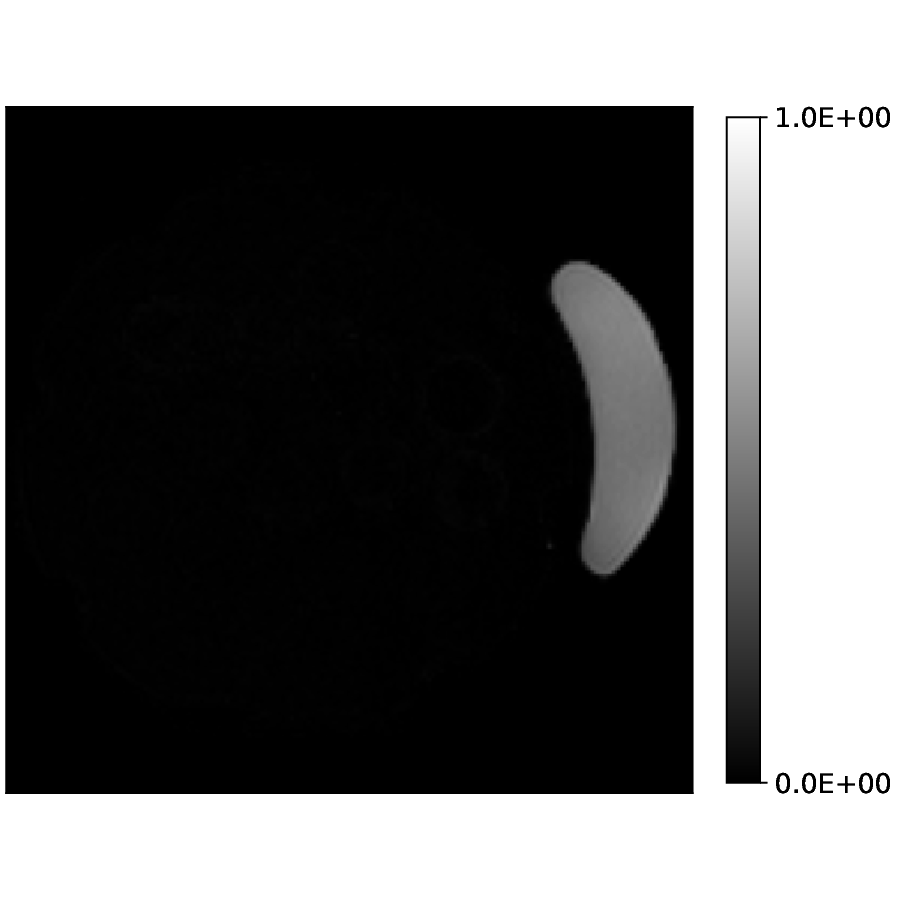}
        \caption{}
        \label{fig:invitro:sRec}
    \end{subfigure}
    \vspace{3pt}
    \begin{subfigure}[t]{0.32\textwidth}
        \includegraphics[width=0.99\textwidth]{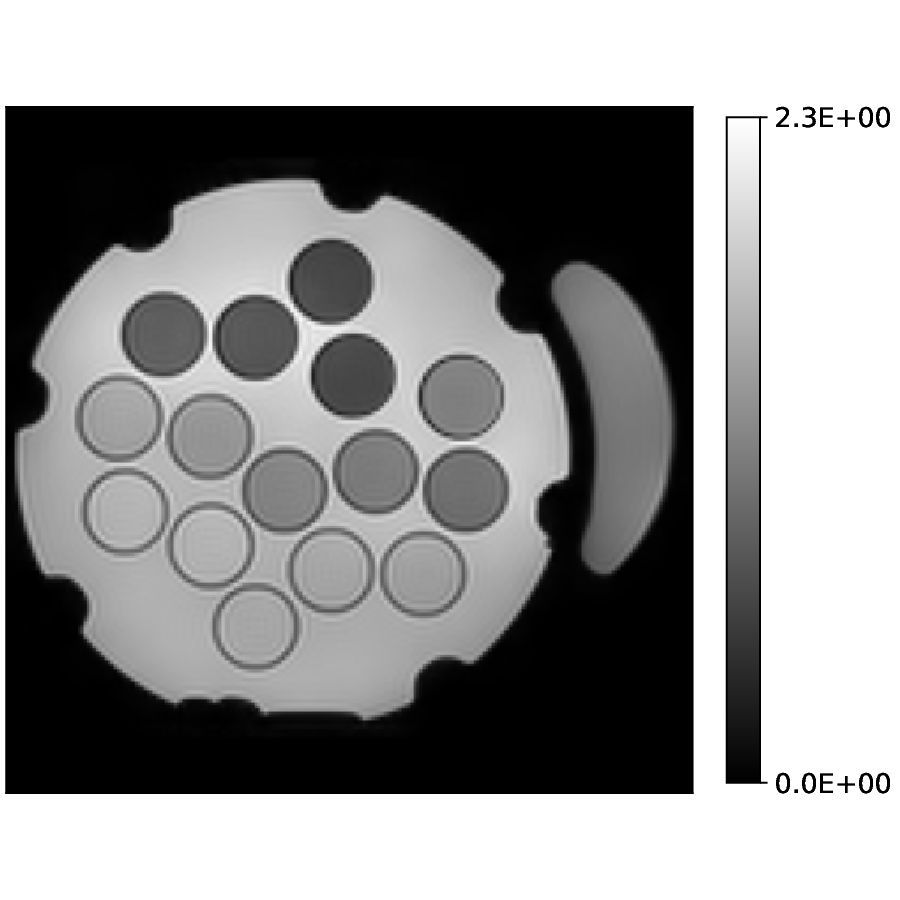}
        \caption{}
        \label{fig:invitro:signalNorm}
    \end{subfigure}%
    \begin{subfigure}[t]{0.32\textwidth}
        \includegraphics[width=0.99\textwidth]{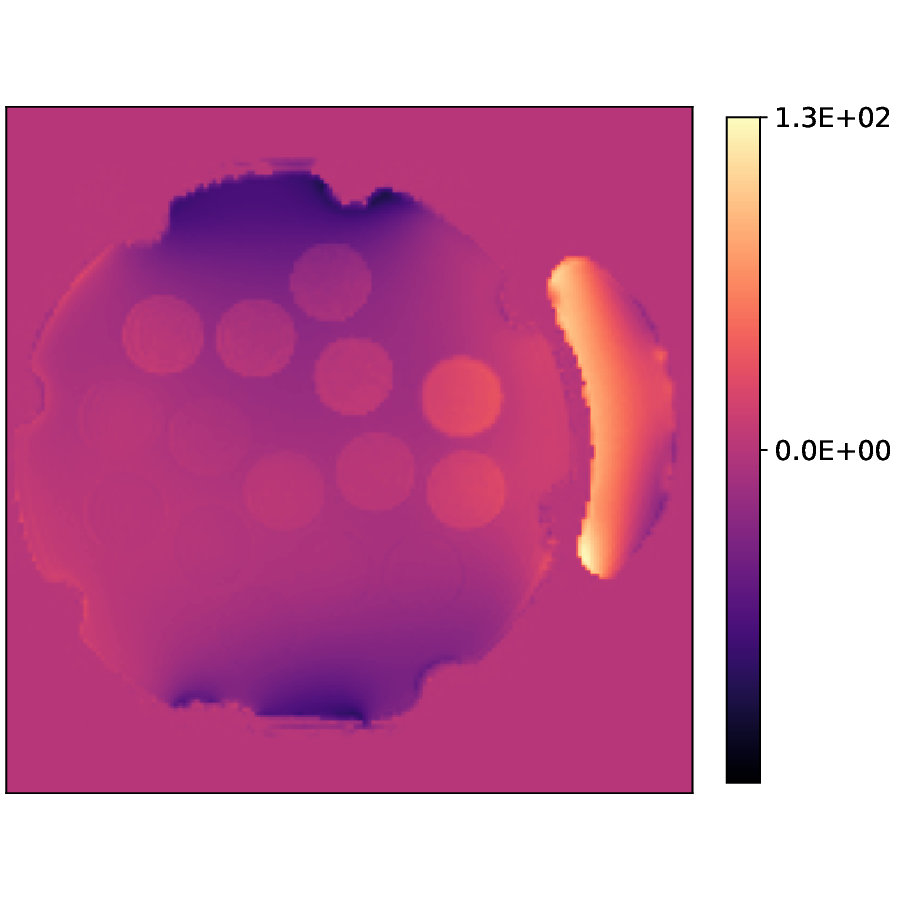}
        \caption{}
        \label{fig:invitro:phRec}
    \end{subfigure}%
    \begin{subfigure}[t]{0.32\textwidth}
        \includegraphics[width=0.99\textwidth]{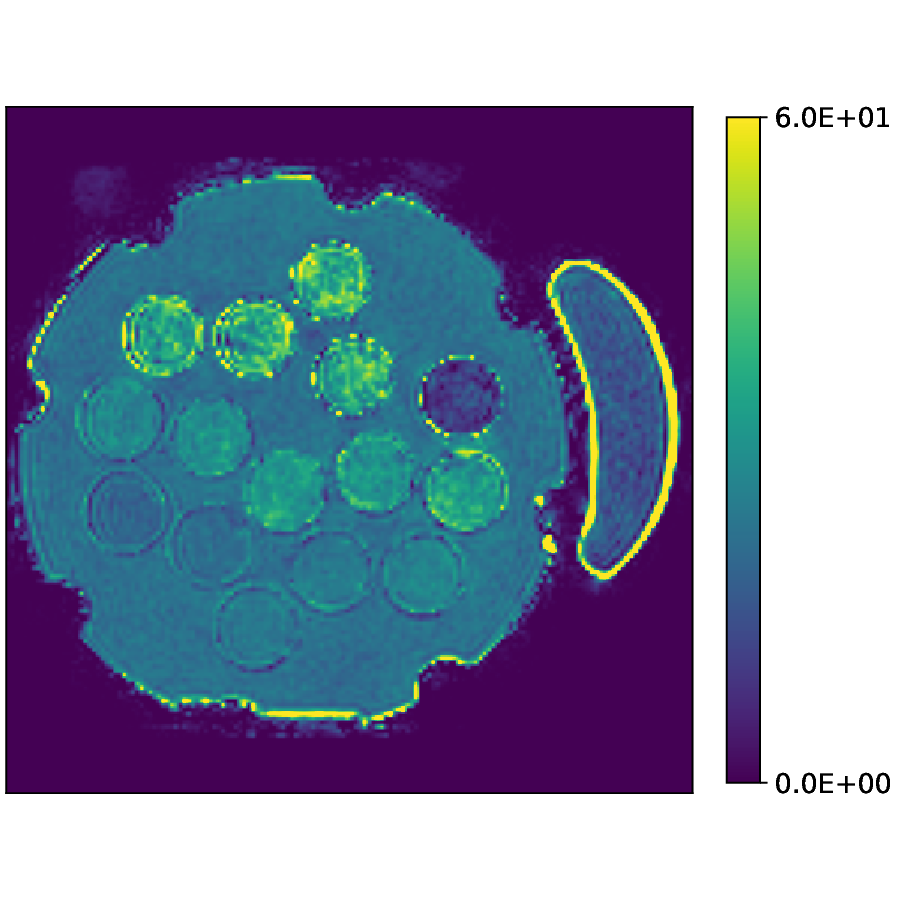}
        \caption{}
        \label{fig:invitro:rRec}
    \end{subfigure}%
    
    \caption{Recovery experiment in an {\em in vitro} water, peanut oil and silicone phantom. (a-c) Recovered concentrations of (a) water, (b) peanut oil and (c) silicone. (d) True signal norm. (e) Recovered fieldmap. (d) Recovered \(r_2^*\) map.
    }
    \label{fig:exp:invitro}
\end{figure*}

\begin{figure*}[!t]
    \centering
    \includegraphics[width=0.40\textwidth]{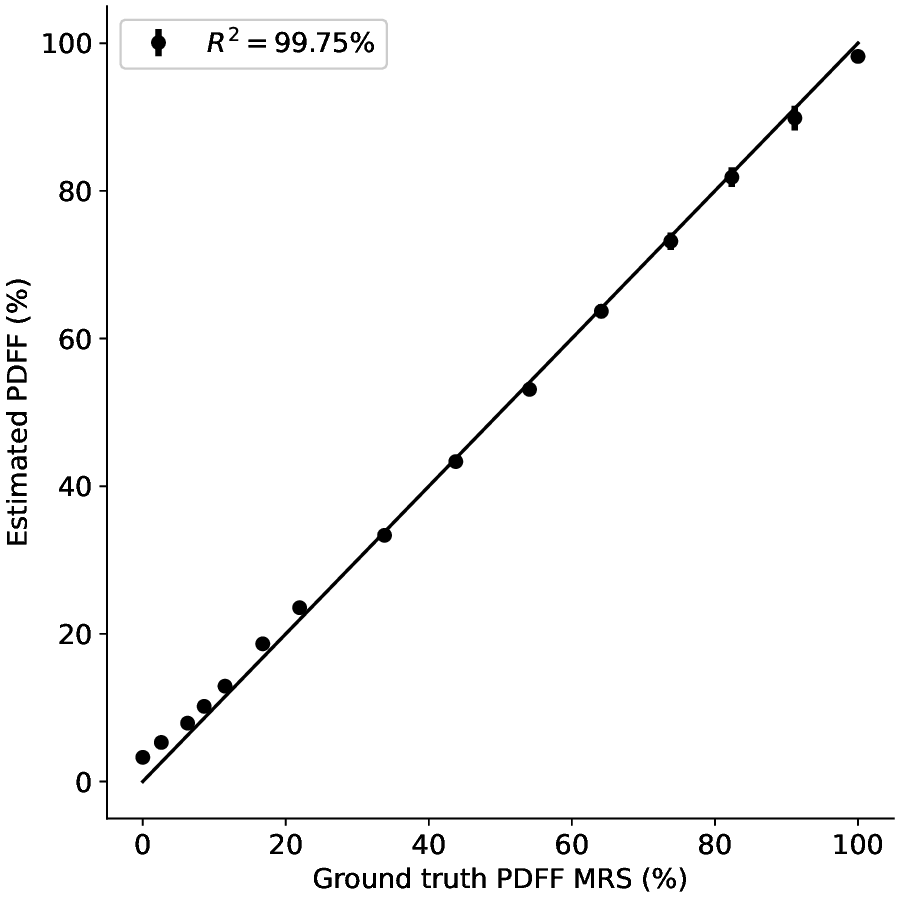}
    \caption{Mean and standard deviation for the error between the recovered PDFF and the reference PDFF obtained with MRS.}
    \label{fig:exp:invitro:error}
\end{figure*}

\section{Discussion}\label{sec:discussion}

Our theoretical results show how the choice of echo times impacts the structure of the solution set. More importantly, this structure is largely independent of the true parameter \(\xio\) and the true concentrations \(\vco\) and can be computed beforehand to determine which concentrations can be recovered exactly or may experience swap artifacts. Furthermore, our results show that the behavior of the residual around the true parameter \(\xio\) can vary significantly. Although the residual has positive curvature near the true parameter, the size of this neighborhood can change substantially depending on the acquisition parameters and the chemical species in the sample. Our results in an {\em in silico} phantom also confirm one of our main findings, namely, that the concentrations and \(r_2^*\) can be recovered accurately even if the fieldmap is not accurate. Our results in an {\em in vitro} phantom shows that our method achieves an accurate quantification of the MRI-PDFF and a correct separation of the three chemical species in the phantom. In this case, the bound of \(20\)Hz on the norm of the gradient was critical as otherwise swap artifacts appeared in the recovered silicone concentration, particularly in regions where the fieldmap is large. To analyze the accuracy and precision of the \(r_2^*\) estimations, a comparison with other algorithms is needed. Although this is part of future experiments, our results already show that accounting for \(r_2^*\) significantly reduces the estimation error and achieves a better recovery of the fieldmap.

\section{Conclusion}\label{sec:conclusion}

In this work we present an analysis of the signal model used in Chemical Shift Imaging. We identify suitable conditions under which exact recovery of the concentrations and \(R_2^*\) is possible, even when the fieldmap cannot be exactly recovered at every voxel. For the imaging problem, we propose a recovery method based on smooth non-convex optimization under convex constraints, and we provide recovery guarantees for this method. As future work, we shall explore how our theoretical findings can assist in the development of novel and robust signal acquisition strategies.

\section*{Acknowledgements}

C.A. was supported by the grant ANID – Fondecyt de Iniciación en Investigación - 11241250. C.SL. was partially funded by grants ANID - Fondecyt - 1211643 and CENIA - FB210017 - Basal - ANID, and an Open Seed Fund between Pontificia Universidad Católica de Chile and the University of Notre Dame. C.SL. thanks the Department for Applied and Computational Mathematics and Statistics at the University of Notre Dame for hosting him during the preparation of this manuscript. 

%% file: proofs.tex
\section{Proof of main results}\label{apx:proofs}

\subsection{Preliminaries}

Let \(m\in \bset{n}\). A {\em selection} is an injective function \(S:\bset{m}\to \bset{n}\). The selection \(S\) thus represents the set \(\set{S(1),\ldots, S(m)}\) and the numbering of its elements. We let \(\sel(m, n)\) the set of selection \(S:\bset{m}\to \bset{n}\). We often manipulate vectors \(\vz\in \C^{2n_s}\). In this case, we use the decomposition
\begin{equation}\label{eq:vectorDecomposition}
    \vz = \bmtx{\vz_1 \\ \vz_2}\quad\mbox{where}\quad \vz_1,\vz_2 \in \C^{n_s}.
\end{equation}
If \(U\subset \C\) is open we say that \(f:U\to \C\) is {\em holomorphic} if it is complex differentiable at every point of \(U\)~\cite[Ch.~1, Sec.~3]{remmert_theory_1991}.

Our results rely on the following two lemmas. The first characterizes the set of zeros of a holomorphic matrix-valued map or a matrix-valued map with real analytic real and imaginary parts.

\begin{lemma}\label{lem:aux:matrixNonSingularSet}
    The following assertions are true. 
    \begin{enumerate}[leftmargin=16pt, label=\roman*.]
        \item{Let \(U\subset \C^n\) be open and connected and let the components of \(\vA : U \to \C^{m\times m}\) be holomorphic. Let \(\mZ_A \subset U\) be the set of \(\vz\) such that \(\vA(\vz)\) is singular. Then either \(\mZ_A = U\) or \(\mZ_A\) is nowhere dense and it has no accumulation points if \(n = 1\).
        }
        \item{Let \(U\subset \R^n\) be open and connected and let the real and imaginary parts of the components of \(\vA : U \to \C^{m\times m}\) be real-analytic. Let \(\mZ_A \subset U\) be the set of \(\vx\) such that \(\vA(\vx)\) is singular. Then \(\mZ_A = U\) or \(\mZ_A\) is nowhere dense.
        }
    \end{enumerate}
\end{lemma}

\begin{proof}[Proof of Lemma~\ref{lem:aux:matrixNonSingularSet}]
    It is apparent in both cases that \(\mZ_A\) is the zero-set of \(\det(\vA)\). The determinant is a polynomial of the components of \(\vA\). Therefore, the first assertion follows from the fact that \(\det(\vA)\) is holomorphic, and thus it cannot vanish on an open set unless it is identically zero~\cite[Thm.~1.19]{range_holomorphic_1986}. Therefore, \(\mZ_A = U\) or it has empty interior. The identity principle yields the claim for \(n=1\)~\cite[Ch.~8, Thm.~1]{remmert_theory_1991}. The second assertion follows from the fact that \(\det(\vA)\) has real-analytic real and imaginary parts and thus they cannot vanish on an open set unless they are identically zero~\cite[Sec.~4.1]{krantz_primer_2002}. Therefore \(\mZ_A\) has empty interior. Furthermore, it has measure zero~\cite{mityagin_zero_2020}.
\end{proof}

The second characterizes the solution set to a system of equations involving complex exponentials that can be reduced to a system of modular equations. Although related to the Chinese remainder theorem, e.g.,~\cite[Thm.~2.2.2]{stein_elementary_2009}, here we are concerned on characterizing the structure of the set of solutions. 

\begin{lemma}\label{lem:aux:modularSystemOfEquations}
    Let \(n\in \N\) and let \(s_1,\ldots, s_n > 0\). Let \(z\in \C\) be such that
    \begin{equation}\label{eq:modularSystemOfEquations}
        k\in \bset{n}:\,\, e^{2\pi i z s_k} = 1.
    \end{equation}
    The following assertions are true.
    \begin{enumerate}[leftmargin=16pt, label=\roman*.]
        \item{If there exists \(k,\ell\in\bset{n}\) such that \(s_k/s_\ell\) is irrational then \(z = 0\) is the only solution to~\eqref{eq:modularSystemOfEquations}.
        }
        \item{If for every \(k,\ell\in \bset{n}\) the quotient \(s_k/s_\ell\) is rational, then 
        \[
            z \in \frac{1}{s_n} \frac{q}{p} \Z
        \]
        where for \(k,\ell\in\bset{n}\) we have that \(s_k / s_n = p_k/q_k\) with \(p_k,q_k\in \Z\) and
        \[
            p = \LCM(\set{p_1,\ldots, p_n})\,\,\mbox{and}\,\, q = \LCM(\set{p q_1 / p_1,\ldots, pq_n/p_n}).
        \]
        }
    \end{enumerate}
\end{lemma}

\begin{proof}[Proof of Lemma~\ref{lem:aux:modularSystemOfEquations}]
    The system~\eqref{eq:modularSystemOfEquations} readily implies that \(z\) must be real. To prove the first assertion, if \(z\) is a solution then there are \(m_k, m_\ell\in \Z\) such that \(z s_k = m_k\) and \(z s_\ell = m_\ell\) whence \(m_k/s_k = m_\ell/s_\ell\). However, this implies that either \(s_k/s_\ell\) or \(s_\ell/s_k\) is rational unless \(m_k = m_\ell = 0\). This implies that \(z = 0\) proving the first assertion. To prove the second assertion, we write \(s_k / s_{n} = p_k / q_k\) for \(p_k,q_k \in \Z\) and \(k\in\bset{m}\). Let \(\bar{p} = \LCM(p_1,\ldots p_{n})\) and define \(\tilde{q}_k = q_k \bar{p}/p_k\). Then \(z \bar{p} t_{n} \in \tilde{q}_k \Z\) for every \(k\in\bset{n}\). We conclude that
    \[
        z \bar{p} s_{n} \in \bigcaps_{k\in\bset{n_e}} \tilde{q}_k \Z = \LCM(\tilde{q}_1,\ldots, \tilde{q}_{n_e}) \Z
    \]
    from where the second assertion follows.
\end{proof}

\subsection{Proof of Lemma~\ref{lem:modelMatrixSubmatrices}}\label{apx:modelMatrixSubmatrices}

Let \(U = I^{n_e}\) and let \(\vPhi: U\to \C^{n_e\times n_s}\) be as in~\eqref{eq:modelMatrix}. For any selection \(S\in \sel(n_s, n_e)\) let \(\vPhi_S: U \to \C^{n_s\times n_s}\) be the map obtained from \(n_s\) rows of \(\vPhi\) with indices \(S(1), \ldots, S(n_s)\), and let \(\mZ_{\Phi_S}\) be as in Lemma~\ref{lem:aux:matrixNonSingularSet}. If \(\mZ_{\Phi} = U\) write \(t_1' = t_{S(1)}, t'_2 = t_{S(2)}, \ldots, t'_{n_s} = t_{S(n_s)}\). Using the Leibniz formula for determinants~\cite[Thm.~9.46]{axler_linear_2024} we have that
\begin{align*}
    0 &= \det(\vPhi_{S}(t',\ldots, t_{n_s})) \\
    &= \vphi_1(t_1')\sums_{\sigma:\, \sigma(1) = 1} \sign(\sigma)\prods_{k=2}^{n_s} \vphi_{k}(t'_{\sigma(k)}) + \sums_{\sigma:\, \sigma(1) \neq 1} \sign(\sigma)\vphi_1(t'_{\sigma(1)})\prods_{k=2}^{n_s} \vphi_{k}(t'_{\sigma(k)}).
\end{align*}
By varying \(t_1'\) and fixing \(t_2',\ldots, t'_{n_s}\) we may represent \(\vphi_1\) as a linear combination of \(\vphi_2,\ldots, \vphi_{n_s}\) contradicting their linear independence. If the factor multiplying \(\vphi_1(t_1')\) is zero, we decompose the remaining sum into the bijections such that \(\sigma(2) = 2\) and \(\sigma(2) \neq 2\). By varying \(t_2'\) and fixing \(t_3',\ldots, t'_{n_s}\) we can \(\vphi_2\) as a linear combination of \(\vphi_3,\ldots, \vphi_{n_s}\). We proceed until this process is exhausted, concluding that \(\mZ_{\Phi}\neq U\). Then \(\mZ_{\Phi_S}\) has measure zero and it follows that 
\[
    \mZ_{\Phi} := \bigcups_{S\in \sel(n_s,n_e)} \mZ_{\Phi_S}
\]
is the countable union of sets of measure zero whence it has measure zero. 

\subsection{Proof of Lemma~\ref{lem:periodicitySetOfW}}\label{apx:periodicitySetOfW}

Let \(n = |\supp(\tvso)|\). By writing \(\set{t_k:\, k\in\supp(\tvso)} = \set{s_1,\ldots, s_n}\) in such a way that \(s_n = \max(\set{t_k:\, k\in\supp(\tvso)})\) and then applying Lemma~\ref{lem:aux:modularSystemOfEquations} the lemma readily follows.

\subsection{Proof of Theorem~\ref{thm:localIdentifiability}}\label{apx:localIdentifiability}

It is apparent that \(\vW:\C \to \C^{n_e\times n_e}\) is an holomorphic map with \(\vW' = 2\pi i \vT \vW\). It follows that \(s: \C\times \C^{n_s}\to \C^{n_e}\) is a holomorphic map of \((\xi,\vc)\) with differential matrix at \((\xio,\vco)\) given by
\begin{equation}\label{eq:differentialMatrix}
    Ds(\xio,\vco) = \bmtx{2\pi i \vT\vW(\xio) \vPhi \vco & \vW(\xio)\vPhi }.
\end{equation}
If it is full-rank then the theorem follows from Corollary~2.6 in~\cite{range_holomorphic_1986}. Since \(n_e \geq 2n_s\) by hypothesis, it suffices to determine sufficient conditions under which its nullspace is trivial.

\begin{lemma}\label{lem:differentialMatrixFullRank}
    The differential matrix~\eqref{eq:differentialMatrix} is full-rank for {\em any} choice of \((\xio,\vco)\) if and only if the matrix \(\vJ := \bmtx{ \vT \vPhi & \vPhi }\) has trivial nullspace. 
\end{lemma}

\begin{proof}[Proof of Lemma~\ref{lem:differentialMatrixFullRank}]
    If \(Ds(\xio,\vco)\) has non-trivial nullspace it is straightforward to verify that \(\vJ\) has non-trivial nullspace. To show the converse, using the decomposition in~\eqref{eq:vectorDecomposition}, if \(\vz\in\ker(\vJ)\) then \(\vT\vPhi \vz_1 + \vPhi\vz_2 = \vnull_{n_e}\). By setting \(\vco = \vz_1\) and \(\delta \vc = 2\pi i \vz_2\) we deduce that
    \[
        2\pi i \vT \vW(\xio) \vPhi \vco + \vW(\xio) \vPhi\delta\vc = 2\pi i\vW(\xio)(\vT \vPhi\vz_1 + \vPhi \vz_2) = \vnull_{n_e}.
    \]
\end{proof}

The condition on \(\vJ\) depends only on the echo times, but not on \((\xio,\vco)\). Let \(\tilde{\vPhi}\) be defined similarly as \(\vPhi\). Then the same arguments used in the proof of Lemma~\ref{lem:modelMatrixSubmatrices} applied to \(\tilde{\vPhi}\) yield a set \(\mZ_{\tilde{\Phi}}\) of measure zero over which \(\vJ\) is not full-rank. Since the hypothesis implies that \(\vphi_1,\ldots,\vphi_{n_s}\) are linearly independenr, for \(\mZ_{\Phi}\) as in Lemma~\ref{lem:modelMatrixSubmatrices}, the set \(\mZ_{\Phi}\cap \mZ_{\tilde{\Phi}}\) also has measure zero, from where the theorem follows. 

\subsection{Proof of Proposition~\ref{prop:recoverableConcentrationSet}, Theorem~\ref{thm:solutionSetStructure} and Proposition~\ref{prop:swaps}}\label{apx:solutionSet}

We prove Proposition~\ref{prop:recoverableConcentrationSet} first. The condition implies that \(\vW(\eta) \tvso = \tvso\). Since \(n_e \geq n_s\) and each \(n_s\times n_s\) submatrix of \(\vPhi\) is non-singular, the support of \(\tvso\) has at least one element. Therefore, \(\imag(\eta) = 0\). If we let \(n = |\supp(\tvso)|\) and write \(\set{t_k:\, k\in\supp(\tvso)} = \set{s_1,\ldots, s_n}\) for \(s_n = \max(\set{t_k:\, k\in\supp(\tvso)})\) then, by applying Lemma~\ref{lem:aux:modularSystemOfEquations}, the proposition follows.

We now prove Theorem~\ref{thm:solutionSetStructure}. Let \(Z_{\Delta}\) be the set of \(\eta\) such that \(\ker(\vDi(\eta))\) is non-trivial. We first prove the following auxiliary result.

\begin{theorem}\label{thm:solutionSetRank}
    Suppose that all \(n_s\times n_s\) submatrices of \(\vPhi\) are non-singular. The following assertions are true.
    \begin{enumerate}[leftmargin=18pt, label=\roman*.]
        \item{We have that \(\dim\ker(\vDi(\eta)) \leq n_s\). 
        }
        \item{If \(n_e \geq n_s + 1\) then \(\dim\ker(\vDi(\eta)) = n_s\) if and only if \(\imag(\eta) = 0\) and there exists an orthonormal basis \(v_1,\ldots, v_{n_s}\) of \(\range(\vPhi)\) independent of \(\eta\), and complex numbers \(e^{i\theta_1},\ldots, e^{i\theta_{n_s}}\) such that \(\vW(\eta)\vv_k = e^{i\th_k} \vv_k\) for \(k\in\bset{n_s}\).
        }
        \item{If \(n_e \geq 2n_s - 1\) and \(\dim\ker(\vDi(\eta)) = n_s\) then \(e^{i\th_1} = \ldots = e^{i\th_{n_s}} = 1\).
        }
        \item{If \(n_e \geq 2n_s\) then \(Z_{\Delta}\) is discrete.
        }
    \end{enumerate}
\end{theorem}

\begin{proof}[Proof of Theorem~\ref{thm:solutionSetRank}]
    To prove the first assertion, let \(m = \dim\ker(\vDi(\eta))\) and let \(\vz_1,\ldots, \vz_m\) be a basis for \(\ker(\vDi(\eta))\). Using the decomposition in~\eqref{eq:vectorDecomposition} note that \(\vz_{2,1},\ldots, \vz_{2,m}\) and \(\vz_{1,1},\ldots, \vz_{1,m}\) must be linearly independent in \(\C^{n_s}\). If not, there exists scalars \(\alpha_1,\ldots, \alpha_m\) such that \(\alpha_1 \vz_{2,1} + \ldots + \alpha_m \vz_{2,m} = \vnull_{n_s}\). However, this implies that
    \[
        \vnull_{n_e} = \vDi(\eta)\left(\sums_{k=1}^m\alpha_k \vz_k\right) =  \vW(\xi) \vPhi \left(\sums_{k=1}^m\alpha_k\vz_{1,k}\right)
    \]
    whence \(\alpha_1 \vz_{1,1} + \ldots + \alpha_m \vz_{1,m} = \vnull_{n_s}\). This contradicts the linear independence of \(\vz_1,\ldots, \vz_m\). The same argument shows that \(\vz_{1,1},\ldots, \vz_{1,m}\) are also linearly independent from where the assertion follows.

    To prove the second assertion, suppose that \(m = n_s\). Then \(\vz_{2,1},\ldots, \vz_{2,m}\) forms a basis for \(\C^{n_s}\) and there exists a \(n_s\times n_s\) non-singular matrix \(\vL\) such that for any \(\vc\in \C^{n_s}\)
    \[
        \vnull_{n_e} = \vW(\eta) \vPhi \vL\vco - \vPhi\vco = (\vW(\eta)\vPhi\vL - \vPhi)\vc.
    \]
    Therefore we have that \(\vW(\eta)\vPhi\vL = \vPhi\) or \(\vW(-\eta)\vPhi = \vPhi\vL\). This implies that \(\vW(-\eta) : \range(\vPhi) \to \range\vPhi)\). Since it is a normal operator it admits an orthonormal basis of eigenvectors~\cite[Thm.~7.31]{axler_linear_2024} which form an orthonormal basis for \(\range(\vPhi)\). If \(\vv\in \C^{n_e}\) is an eigenvector with eigenvalue \(\lambda\in\C\) it follows that \(e^{-2\pi i \eta t_k} = \lam\) for \(k\in \supp(\vv)\). Since every \(n_s\times n_s\) submatrix of \(\vPhi\) is non-singular, there are at least \(n_e - n_s + 1\) equations, and, since \(n_e > n_s\), there are at least two such equations. This implies that \(\imag(\eta) = 0\) and that \(|\lam| = 1\). Hence, the eigenvalues of \(\vW(-\eta)\) have the form \(e^{-i\th_1},\ldots, e^{-i\th_{n_s}}\) as we wanted to show. Conversely, if there exists an orthonormal basis as in the statement, then we can define \(\vu_1,\ldots,\vu_{n_s} \in \C^{n_s}\) through the relation \(\vv_k = \vPhi \vu_k\) for \(k\in \bset{n_s}\). It is clear that \(\vu_1,\ldots, \vu_{n_s}\) are a basis for \(\C^{n_s}\) and that the collection \(\vz_1,\ldots, \vz_{n_s}\) where \(\vz_{k,1} = e^{-i\th_k}\vu_k\) and \(\vz_{k,2} = -\vu_k\) is a basis for \(\ker(\vPhi)\). Hence, \(m = n_s\). To prove that the orthonormal basis \(\vv_1,\ldots, \vv_{n_s}\) of \(\range(\Phi)\) is uniquely determined, it suffices to see that if \(m = n_s\) for \(\eta_1,\eta_2\) the previous arguments show that \(\vW(-\eta_1),\vW(-\eta_2):\range(\vPhi)\to \range(\vPhi)\). Since they commute, they are jointly diagonalizable~\cite[Thm.~5.76]{axler_linear_2024}. 
    
    To prove the third assertion, if \(\vv_k, \vv_\ell\) are two eigenvectors of \(\vW(-\eta)\) then their supports have at least \(n_e - n_s + 1\) elements. Since \(n_e \geq 2n_s - 1\) this implies that the intersection of their supports has at least 1 element. If the corresponding eigenvalues are \(e^{-i\th_k}\) and \(e^{-i\th_\ell}\) and the common index is \(p\) then \(e^{-2\pi i \eta t_p} = e^{-i\th_k} = e^{-i\th_\ell} = e^{-2\pi i \eta t_p}\) whence all eigenvalues must be the same.

    To prove the fourth assertion, for any \(S\in \sel(2n_s, n_e)\) let \(\vDi_S: U \to \C^{2n_s\times 2n_s}\) be the map obtained from the \(2n_s\) rows of \(\vPhi\) with indices \(S(1), \ldots, S(2n_s)\). Since \(\eta\mapsto \vDi(\eta)\) is holomorphic we may apply Lemma~\ref{lem:aux:matrixNonSingularSet}. The same arguments used in the proof of Lemma~\ref{lem:modelMatrixSubmatrices} show that if \(\mZ_{\vDi_S} = U\) then by expanding the determinant we may represent an exponential of the form \(e^{2\pi i \eta s_1}\) as a linear combination of exponentials of the form \(e^{2\pi i \eta s_2},\ldots, e^{2\pi i \eta s_{2 n_s}}\) and a constant, where \(s_1, \ldots, s_{2n_s}\) are sums of subsets of the echo times. Since these are linearly independent, we must have that \(\mZ_{\vDi_S} = U\) is a discrete set with no accumulation points. Using the fact that \(Z_{\Delta}\subset \mZ_{\vDi_S}\) the assertion follows.
\end{proof}

To prove Theorem~\ref{thm:solutionSetStructure} we first
decompose \(Z_{\Delta}\) into
\[
    Z_{\Delta}^* := \set{\eta\in Z_{\Delta}:\, \dim\ker(\vDi(\eta)) = n_s}
\]
and \(Z_{\Delta}^{\circ} = Z_{\Delta} \setminus Z_{\Delta}^*\). Define
\[
    \mZ_c = \bigcups_{\eta\in Z_{\Delta}^{\circ}}\, Z_c(\eta).
\]
Since \(n_e \geq 2n_s\) we have that \(Z_{\Delta}\) is discrete by the fourth assertion in Theorem~\ref{thm:solutionSetRank}. On one hand, since each one of the subspaces in the above union has dimension strictly less than \(n_s\) we conclude that \(\mZ_c\) is the countable union of sets of measure zero and thus it has measure zero. On the other, if \(\eta \in Z_{\Delta}^{\circ}\) then \(\dim\ker(\vDi(\eta)) = n_s\). However, using the decomposition in~\eqref{eq:vectorDecomposition}, the third assertion of Theorem~\ref{thm:solutionSetRank} implies that \(\vz\in \ker(\vDi(\eta))\) yields \(\vz_{1} = -\vz_{2}\). Hence, if \(\vco\in \mZ_c^*\) we have that \(Z_{\Delta}^* = \Xi(\vco)\). This proves the theorem.

Finally, we prove Proposition~\ref{prop:swaps}. In this case, it suffices to observe that in the previous argument, if \(\vco\in \mZ_c^*\) then we can use the second assertion in Theorem~\ref{thm:solutionSetRank}. This yields the proposition.

\subsection{Proof of Theorem~\ref{thm:localExactRecovery}}\label{apx:localExactRecovery}

Before proving Theorem~\ref{thm:localExactRecovery} we provide some auxiliary results about the regularity of the residual matrix in Section~\ref{apx:residualMatrixRegularity} and the regularity of the residual in Section~\ref{apx:residualRegularity}. We then proceed to prove the theorem in Section~\ref{apx:localConvergence}. 

\subsubsection{Regularity of the residual matrix}\label{apx:residualMatrixRegularity}

We summarize the main results about the regularity of \(\vR\) in the following lemma.
\begin{lemma}\label{lem:residualMatrixDerivatives}
    Let \(n\in\N\). Then for any \(\xi\in \C\) we have that \(\vR^{(n)}(\xi)\adj = \vR^{(n)}(\xi\adj)\), that \(\vR^{(n)}(\xi) = 2\pi i \bprod{\vT}{\vR^{(n-1)}(\xi)}\) and that \(\nrmop{\vR^{(n)}(\xi)} \leq  \tne^n  e^{\tS |\imag(\xi)| /2}\)
\end{lemma}

\begin{proof}[Proof of Lemma~\ref{lem:residualMatrixDerivatives}]
    The lemma follows from an induction argument on \(n\). We prove only the initial case and leave the details to the reader. Since \(\vW'(\xi) = 2\pi i \vT \vW(\xi)\) we have that
    \begin{align*}
        \vR'(\xi) &= 2\pi i \vT \vW(\xi) \vPp\vW(-\xi) - 2\pi i \vW(\xi) \vPp \vW(-\xi) \vT \\
        &= 2\pi i (\vT\vR(\xi) - \vR(\xi) \vT) \\
        &= 2\pi i \bprod{\vT}{\vR(\xi)}.
    \end{align*}
    Similarly,
    \[
        \vR(\xi)\adj = \vW(-\xi)\adj \vPp\adj\vW(\xi)\adj = \vW(\xi\adj)\vPp\vW(-\xi\adj) = \vR(\xi\adj).
    \]
    Finally, suppose that \(\imag(\xi) \geq 0\). Then
    \[
        \nrm{\vR(\xi)\vs}_2 = \nrm{\vW(\xi)\vPp\vW(-\xi)\vs}_2 \leq e^{-2\pi t_1\imag(\xi)} \nrmop{\vPp}\nrm{\vW(-\xi)\vs}_2 \leq e^{2\pi (t_{n_e}-t_1) \imag(\xi)}
    \]
    where we used the fact that \(\nrmop{\vP_{\perp}} = 1\). If \(\imag(\xi) < 0\) it follows that
    \[
        \nrm{\vR(\xi)\vs}_2 = \nrm{\vW(\xi)\vPp\vW(-\xi)\vs}_2 \leq e^{2\pi t_{n_e}|\imag(\xi)|} \nrmop{\vPp}\nrm{\vW(-\xi)\vs}_2 \leq e^{2\pi (t_{n_e}-t_1) |\imag(\xi)|}.
    \]
\end{proof}

\subsubsection{Regularity of the residual}\label{apx:residualRegularity}

To simplify the notation, for \(n\in\No\) let
\[
    \psi_n(\xi,\xi\adj) = \frac{1}{2}\iprod{\vso}{\vR^{(n)}(\xi\adj)\vR^{(n)}(\xi)\vso}.
\]
The results about their regularity are summarized below.

\begin{lemma}\label{lem:residualsDerivativeBound}
    For \(n\in \No\) we have that
    \[
        |\psi_n(\xi + \eta, \xi\adj + \eta\adj) - \psi_n(\xi,\xi\adj)| \leq \tne^{2n+1}\nrm{\vso}^2  |\eta| \beta(\tS \imag(\xi), \tS\imag(\eta))
    \]
    where
    \begin{equation}\label{eq:betaFunction}
        \beta(a, b) := \int_0^1 e^{|a + \th b|} d\th
    \end{equation}
\end{lemma}

\begin{proof}[Proof of Lemma~\ref{lem:residualsDerivativeBound}]
    It is apparent that \(\psi_n\) is real-valued and Wirtinger differentiable with
    \[
        \dxi \psi_n(\xi,\xi\adj) = \frac{1}{2}\iprod{\vR^{(n)}(\xi)\vso}{\vR^{(n+1)}(\xi)\vso}.
    \]
    Using the fundamental theorem of calculus
    \[
        \psi_n(\xi + \eta, \xi\adj + \eta\adj) - \psi_n(\xi, \xi\adj) = 2\real\left(\eta \int_0^1 \dxi \psi_n(\xi+\th\eta, \xi\adj+\th\eta\adj)\, d\th\right).
    \]
    Lemma~\ref{lem:residualMatrixDerivatives} yields the bound
    \begin{align*}
        \left|\psi_n(\xi + \eta, \xi\adj + \eta\adj) - \psi_n(\xi, \xi\adj)\right| &\leq |\eta|\nrm{\vso}^2 \int_0^1 \nrmop{\vR^{(n)}(\xi+\th\eta)}\nrmop{\vR^{(n+1)}(\xi+\th\eta)} d\th\\
        &\leq |\eta|\tne^{2n+1}\nrm{\vso}^2 \int_0^1 e^{\tS |\imag(\xi) + \th\imag(\eta)|} d\th
    \end{align*}
    from where the lemma follows.
\end{proof}

The function \(\beta :\R\times \R\to\R\) is independent of the parameters of the problem. When \(a + b \geq 0\) it has the closed form \(\beta(a, b) = b^{-1}e^{a} (e^{b} - 1) = e^a e^{b/2} \sinh(b/2)\).

\subsubsection{Bounds on the radius of monotonicity}\label{apx:radiusOfMonotonicity}

Using Lemma~\ref{lem:residualMatrixDerivatives} we can compute the Hessian \(\nabla_W^2 \fo\) of the residual in closed form. The quadratic form it induces on the variable \((\eta,\eta\adj)\) is represented as
\begin{equation}\label{eq:wirtingerHessianQuadAction}
    H_W\fo(\xi,\xi\adj)(\eta,\eta\adj) = |\eta|^2 \nrm{\vR'(\xi)\vso}_2^2 + \real(\eta^2 \iprod{\vso}{\vR(\xi\adj)\vR''(\xi)\vso}).
\end{equation}
If \(\xio\) is the true parameter then~\eqref{eq:hessianAtTrueParameter} holds as \(\vR(\xio)\vso = 0\). 

\begin{lemma}\label{lem:hessianAtTrueParameterPositive}
    Suppose that Theorem~\ref{thm:localIdentifiability} holds. If \(\vco \neq 0\) then \(\vR'(\xio)\vso\neq 0\).
\end{lemma}

\begin{proof}[Proof of Lemma~\ref{lem:hessianAtTrueParameterPositive}]
    If \(\vR'(\xio)\vso = 0\) then Proposition~\ref{lem:residualMatrixDerivatives} implies that \(\vR(\xio)\vT\vso = 0\). However, it follows that \(\vT\vPhi\vco = \vPhi\vc\) for some \(\vc\). Since Theorem~\ref{thm:localIdentifiability} holds, from Lemma~\ref{lem:differentialMatrixFullRank} we conclude that \(\vc = \vco = 0\). This contradiction yields the lemma.
\end{proof}

To bound the radius of the neighborhood on which the Hessian is positive definite, which we call the {\em radius of monotonicity}, we determine the radius over which the lower bound
\[
    \HWfo(\xi,\xi\adj)(\eta,\eta\adj) \geq \nrm{\vR'(\xi)\vso}_2^2 - \nrm{\vR(\xi)\vso}_2\nrm{\vR''(\xi)\vso}_2,
\]
valid for \(|\eta| = 1\), remains positive. We have the following lemma.

\begin{lemma}\label{lem:radiusOfMonotonicity}
    Let \(\rho \in (0, 1)\). Then for any \(\xi\in \Hp\) such that
    \[
        |\xi - \xio| \leq \frac{1}{\tS} W\left(\frac{\tS e^{-\tS \imag(\xio)}}{2\nrm{\vso}_2^2}\gxio^+(\rho)\right)
    \]
    we have that
    \[
        \rho \HWfo(\xio,\xio\adj)(\eta, \eta\adj) \leq \HWfo(\xio,\xio\adj)(\eta, \eta\adj)\leq (2 + \rho)\HWfo(\xio,\xio\adj)(\eta, \eta\adj)
    \]
    for \(|\eta| = 1\).
\end{lemma}

\begin{proof}
    For simplicity, define
    \[
        \gxio(\eta) := 2\nrm{\vso}^2|\eta|\beta(\tS\imag(\xio), \tS\imag(\eta)).
    \]
    From Lemma~\ref{lem:residualsDerivativeBound} we deduce the bounds
    \begin{align*}
        \nrm{\vR'(\xio + \eta)\vso}^2 &\geq \nrm{\vR'(\xio)\vso}^2 -  \tne^{3}\gxio(\eta)\\
        \nrm{\vR(\xio + \eta)\vso}^2 &\leq  \tne\gxio(\eta)\\
        \nrm{\vR''(\xio + \eta)\vso}^2 &\leq \nrm{\vR''(\xio)\vso}^2 +  \tne^{5}\gxio(\eta)
    \end{align*}
    where we used the fact that \(\vR(\xio)\vso = 0\). Using the subaditivity of the square-root on the non-negative reals we deduce the bound
    \begin{multline}\label{eq:hessianLowerBoundProof}
        \nrm{\vR'(\xio+\eta)\vso}_2^2 - \nrm{\vR(\xio+\eta)\vso}_2\nrm{\vR''(\xio+\eta)\vso}_2 \geq \\ 
        \nrm{\vR'(\xio)\vso}_2^2- \tne^3 \gxio(\eta) - (\tne\gxio(\eta))^{1/2}\nrm{\vR''(\xio)\vso}_2 - (\tne \gxio(\eta))^{1/2}(\tne^{5}\gxio(\eta))^{1/2} = \\
        \nrm{\vR'(\xio)\vso}_2^2 - (\tne\gxio(\eta))^{1/2}\nrm{\vR''(\xio)\vso}_2 - 2\tne^3\gxio(\eta) \geq \rho \nrm{\vR'(\xio)\vso}_2^2
    \end{multline}
    where we used~\eqref{eq:hessianAtTrueParameter}. The values of \(\eta\) for which the lower bound is positive must solve the inequality
    \[
        -\frac{1-\rho}{2}\frac{\nrmop{\vR'(\xio)\vso}^2}{\tne^3} + \frac{1}{2}\frac{\nrm{\vR''(\xio)\vso}}{\tne^{5/2}}\gxio(\eta)^{1/2} + \gxio(\eta) \leq 0.
    \]
    The positive root for this equation is precisely \(\gxio^+(\rho)\). Therefore, for the lower bound to be satisfied, the inequality 
    \[
        |\eta|\beta(\tS\imag(\xio), \tS\imag(\eta)) \leq \frac{1}{2\nrm{\vso}_2^2} \gxio^+(\rho)
    \]
    must be satisfied. Since the upper bound is positive, and the left-hand side is a continuous function of \(\eta\) equal to zero when \(\eta =0\) we conclude there is a small neighborhood around the origin for which the above holds. If we further assume that \(\xio + \eta \in \Hp\) then we may use an explicit expression for \(\beta\) to deduce that 
    \[
        \beta(\tS\imag(\xio), \tS\imag(\eta)) = e^{\tS \imag(\xio)} \int_0^1 e^{\th \tS \imag(\eta) }\,d\th \leq e^{\tS \imag(\xio)} e^{\tS |\eta| }. 
    \]
    Thus,
    \[
        \tS |\eta| e^{\tS |\eta|} \leq \frac{\tS e^{-\tS\imag(\xio)}}{2\nrm{\vso}_2^2} \gxio^+(\rho)
    \]
    from where the lower bound follows by identifying the Lambert \(W\) function. The upper bound follows from the fact that~\eqref{eq:hessianLowerBoundProof} implies that
    \[
        \nrm{\vR(\xio+\eta)\vso}_2\nrm{\vR''(\xio+\eta)\vso}_2 \leq (1 + \rho) \nrm{\vR'(\xio)\vso}_2^2
    \]
    from where a straightforward argument yields the bound.
\end{proof}

The exact same arguments yield the tighter bound implicit in the inequality
\begin{equation}\label{eq:radiusOfMonotonicityBound:loose}
    |\eta|e^{\tS\imag(\eta)/2}\sinh(\tS\imag(\eta)/2) \leq \frac{e^{-\tS \imag(\xio)}}{2\nrm{\vso}_2^2} \gxio^+(\rho)
\end{equation}
for any \(\xio + \eta \in \Hp\). A refined estimate yielding an implicit, but tighter bound can be found as follows. Since \(\xio\) is the true parameter, we can use the second-order expansion
\begin{align*}
    \frac{1}{2}\nrm{\vR(\xio + \eta)\vso}^2_2 &= \frac{1}{2} \real\left(\int_0^1 (\eta,\eta\adj)\nabla_W^2 f(\xio + \theta\eta, \xio\adj + \theta\eta\adj)(\eta,\eta\adj)d\theta  \right)\\
    &\leq \frac{1}{2}|\eta|^2 \int_0^1 (\nrm{\vR'(\xio + \th\eta)\vso}_2^2 + \nrm{\vR(\xio + \th\eta)\vso}_2 \nrm{\vR''(\xio + \th\eta)\vso}_2) d\theta\\
    &\leq \tne^2|\eta|^2\nrm{\vso}_2^2\beta(\tS \imag(\xio), \tS \imag(\eta))\\
    &= \frac{1}{2}\tne^2 |\eta|\gxio(\eta).
\end{align*}
In this case, we obtain
\begin{multline}\label{eq:radiusOfMonotonicityBound:tighter}
    \nrm{\vR'(\xio+\eta)\vso}_2^2 - \nrm{\vR(\xio+\eta)\vso}_2\nrm{\vR''(\xio+\eta)\vso}_2 \geq \\ 
    \nrm{\vR'(\xio)\vso}_2^2- \tne^3 \gxio(\eta) - (\tne^2|\eta|\gxio(\eta))^{1/2}\nrm{\vR''(\xio)\vso}_2 - (\tne^2|\eta|\gxio(\eta))^{1/2}(\tne^{5}\gxio(\eta))^{1/2} = \\
    \nrm{\vR'(\xio)\vso}_2^2 - \tne|\eta|^{1/2}\gxio(\eta)^{1/2}\nrm{\vR''(\xio)\vso}_2 - \tne^3(1+ |\eta|^{1/2})\gxio(\eta) \geq \rho  \nrm{\vR'(\xio)\vso}_2^2.
\end{multline}
However, the lower bound is no longer a polynomial, and thus numerical methods must be used to estimate the radius.

\subsubsection{Local convergence}\label{apx:localConvergence}

We can now prove Theorem~\ref{thm:localExactRecovery}. Let \(\eta := \xi - \xio\). Suppose that \(|\eta| \leq r_0\). From
\[
    \nabla_W \fo(\xi,\xi\adj) - \nabla_W \fo(\xio,\xio\adj) = \int_0^1 \nabla^2_W \fo(\xio + \th\eta, \xio\adj + \th\eta\adj)(\eta, \eta\adj) d\th
\]
and the upper bound in Lemma~\ref{lem:radiusOfMonotonicity} we deduce that
\[
    \nrmW{\nabla_W \fo(\xi,\xi\adj) - \nabla_W \fo(\xio,\xio\adj)} \leq \Lo \nrmW{(\eta, \eta\adj)}
\]
for
\[
    \Lo \geq (2 + \rho)\nrm{\vR'(\xio)\vso}_2^2.
\]
Similarly,
\[
    (\nabla_W \fo(\xi,\xi\adj)\adj - \nabla_W \fo(\xio,\xio\adj)\adj)(\eta, \eta\adj) = \int_0^1 \HWfo(\xio + \th\eta, \xio\adj + \th\eta\adj)(\eta, \eta\adj) d\th \geq \muo\nrmW{(\eta,\eta\adj)}^2
\]
for
\[
    \muo \leq \rho\nrm{\vR'(\xio)\vso}_2^2.
\]
Consequently,
\begin{align*}
    \nrmW{(\xi, \xi\adj)^+ - (\xio,\xio\adj)}^2 &= \nrmW{(\xi, \xi\adj) - (\xio, \xio\adj) - \alpha (\nabla_W \fo(\xi,\xi\adj)\adj - \nabla_W \fo(\xio,\xio\adj)\adj)}^2\\
    &= \nrmW{(\xi - \xio,\xi\adj - \xio\adj)}^2 \\
    &\quad -\: 2\alpha (\nabla_W \fo(\xi,\xi\adj)\adj - \nabla_W \fo(\xio,\xio\adj)\adj)(\xi - \xio,\xi\adj-\xio\adj) \\
    &\quad +\:\alpha^2\nrmW{\nabla_W \fo(\xi,\xi\adj)\adj - \nabla_W \fo(\xio,\xio\adj)\adj}^2\\
    &\leq (1 - 2\muo \alpha + \Lo \alpha^2)\nrmW{(\xi - \xio,\xi\adj -\xio\adj)}^2
\end{align*}
where we used the fact that \(\nabla_W \fo(\xio,\xio\adj) = 0\)~(cf. the proof of Theorem~2.1.14 in~\cite{nesterov_lectures_2018}). The distance is strictly decreasing if \(1 - 2\muo \alpha + \Lo \alpha^2 < 1\) whence \(\alpha < 2\muo/\Lo\) proving the theorem. The particular bound provided is obtained by replacing the bounds for \(\muo\) and \(\Lo\) in the above.

\subsection{Proof of Theorem~\ref{thm:localRecoveryImaging}}\label{apx:localRecoveryImaging}

Suppose that there exists another \(\phi\in \Cg\) that is a global minimizer. By Theorem~\ref{thm:solutionSetStructure} we have that
\[
    v\in V:\,\, \phi(v) \in \phio(v) + \Xi(\vco).
\]
Let \(V_0\subset V\) be the set of voxels on which \(\phi\) and \(\phio\) coincide, which is non-empty by hypothesis. By definition of \(r\) we have that
\[
    v\in V\setminus V_0:\,\, |\phi(v) - \phio(v)| \geq r.
\]
Let \(v\in V\setminus V_0\) be any point with a neighbor \(\vo\in V_0\). We consider two cases. First, suppose that \(v - \vo = e_i\) where \(e_i\) is the vector with its \(i\)-th component equal to one and the remaining ones equal to zero. Then
\[
    \epsg(\vo) \geq |\phi(\vo + e_i) - \phi(\vo)| = |\phi(\vo + e_i) - \phio(\vo + e_i) + \phio(\vo + e_i) - \phio(\vo)| \geq r - \epsg(\vo)
\]
which is not possible as \(\epsg(v) < r/2\) for any \(v\in V\). Second, suppose that \(v - \vo = -e_i\). In this case we have that \(e_i = \vo - v\) and we write
\[
    \epsg(v) \geq |\phi(v + e_i) - \phi(v)| = |\phio(v + e_i) - \phio(v) + \phio(v) - \phi(v)| \geq r - \epsg(\vo)
\]
which is once again a contradiction. The claim then follows. 